\documentclass[]{scrartcl}

\usepackage{graphicx}%
\usepackage{multirow}%
\usepackage{amsmath,amssymb,amsfonts}%
\usepackage{amsthm}%
\usepackage{mathrsfs}%
\usepackage[title]{appendix}%
\usepackage{xcolor}%
\usepackage{textcomp}%
\usepackage{manyfoot}%
\usepackage{booktabs}%
\usepackage{algorithm}%
\usepackage{algorithmicx}%
\usepackage{algpseudocode}%
\usepackage{listings}%
\usepackage{pdfpages}
\usepackage{csquotes}
\usepackage[
backend=biber,
style=authoryear,
url=false,
giveninits=true,
isbn=false,
maxnames=2
]{biblatex}
\DeclareNameAlias{author}{family-given}
\newtheorem{theorem}{Theorem}

\newtheorem{lemma}[theorem]{Lemma}
\theoremstyle{remark}
\newtheorem{example}{Example}
\newtheorem{remark}{Remark}

\usepackage{dsfont}
\usepackage{mathtools}

\newtheorem{condition}{Condition}

\newcommand{\Cov}{\operatorname{Cov}}
\renewcommand{\P}{\mathds P}
\usepackage{mathrsfs}
\newcommand{\E}{\mathds E}
\newcommand{\R}{\mathbb R}
\newcommand{\N}{\mathbb N}
\newcommand{\diff}{~\!\mathrm d}
\DeclareMathOperator{\GPD}{GPD}
\DeclareMathOperator{\indic}{{\mathds 1}}
\usepackage[colorlinks=true,linkcolor=purple,citecolor=cyan,urlcolor=magenta,pdfborder={0 0 0}]{hyperref}
\usepackage{orcidlink}

\begin{document}

\title{Extrapolating into the Extremes with Minimum Distance Estimation}   
\author{
Alexis Boulin\thanks{Ruhr-Universität Bochum, Fakultät für Mathematik. Email: \href{mailto:alexis.boulin@rub.de}{alexis.boulin@rub.de}} \orcidlink{0000-0003-0548-2726
}
\and
Erik Haufs\thanks{Ruhr-Universität Bochum, Fakultät für Mathematik. Email: \href{mailto:erik.haufs@rub.de}{erik.haufs@rub.de}} \orcidlink{0009-0008-8194-7445}
}

\date{\today}

\maketitle
\begin{abstract}
    Understanding complex dependencies and extrapolating beyond observations are key challenges in modeling environmental space-time extremes. To address this, we introduce a simplifying approach that projects a wide range of multivariate exceedance problems onto a univariate peaks-over-threshold problem. In this framework, an estimator is computed by minimizing the $L_2$-distance between the empirical distribution function of the data and the theoretical distribution of the model. Asymptotic properties of this estimator are derived and validated in a simulation study. We evaluated our estimator in the EVA (2025) conference Data Challenge as part of Team Bochum's submission. The challenge provided precipitation data from four runs of LENS2, an ensemble of long-term weather simulations, on a $5 \times 5$ grid of locations centered at the grid point closest to Asheville, NC. Our estimator achieved a top-three rank in two of six competitive categories and won the overall preliminary challenge against ten competing teams.
\end{abstract}

\smallskip
\smallskip
\smallskip


Secondary
62F12, 
62P12. 

\tableofcontents

\section{Introduction}\label{sec:Intro}

Extreme value statistics is frequently concerned with the problem of estimating some extreme occurrence probability. Classical models are often limited to either a spatial \textit{or} a temporal analysis of an extreme event.

In environmental statistics, however, the interest lies in the spatio-temporal extent of an extreme event: for example, the severeness of a flooding is influenced by the cumulative precipitation over different time scales (hourly to several days) as well as the regional spread of extreme rain. The tasks in the EVA 2025 data challenge try to capture such complexities.
All posed tasks may be phrased in the quite general setting of estimating exceedance probabilities. They have been designed in a way to incorporate multivariate methods, targeting at the simultaneous exceedance probability of precipitation at multiple distinct locations on a grid, as well as capturing the temporal dependence structure by asking for exceedances over consecutive days. A detailed description of the challenge design is provided in \cite{EVA2025DataChallenge}. 
The data challenge consisted of a preliminary and a competition challenge, each subdivided into three target probabilities, which will be denoted by (P1) -- (P3) and (C1) -- (C3) in the following. The targets refer to a simulated data set of daily precipitation at 25 locations, organized into a $5\times 5$ grid. At each location, 165 years of daily observations are recorded; four simulation runs of this model were presented.      

Our contribution may be seen as two-fold, first methodological, and secondly, theoretical. On the methodological level, we provide a strategy how to project all, up to 25-dimensional, problems into a simplifying univariate framework, which we exploit to estimate the target probabilities. Within the univariate setting, an $L_2$-minimum distance estimation is employed to extrapolate from less extreme threshold exceedances to the target exceedance probabilities. Latter idea emerged from a diagnostic tool frequently used in the evaluation of multivariate extreme value methods to compare theoretical survival probabilities to empirical ones, see, e.g., \cite[Fig. 3]{kiriliouk2022estimating}, \cite[Fig. 7]{Li2024} or \cite[Fig. 4]{buritica2025modeling}.

Our proposed minimum distance approach directly fits a parametric model to the empirical survival probabilities and uses this model for extrapolation, similar to a peaks-over-threshold approach. Latter is widely applied throughout extreme value statistics, with GPD modeling dating back as far as \cite{Pickands1975} and with statistical properties studied by \cite{Balkema1974,Davison1984,Davison1990}, among many others.

Our major theoretical contribution lies in the introduction of an $L_2$ Minimum-Distance Estimator, for which we provide asymptotic theory as well as an exhaustive efficiency comparison to maximum likelihood estimation.

Minimum distance estimation strategies are widely used in different areas of statistics \cite{Drossos1980,Clarke1994,Ozturk1997}, increasingly popular in situations demanding robust estimation. Notable contributions to the $L_2$-MDE may be attributed to, for instance, \cite{Hettmansperger1994}. Even within extreme value statistics, they are a commonly applied tool, see for instance \cite{Berghaus2013, broadwater2009adaptive, Dietrich01011996, Yilmaz18112021}. In particular for the GPD, the MLE has been long criticized to yield non-robust estimates and attempts have been made to provide more robust estimation strategies, such as \cite{Jurez2004}.

Probably closest to our proposed approach is the work of \cite{Chen02102017}. However, they only minimize their distance on the GPD sample and do not minimize a distance related to the whole empirical survival functions. Besides, we prove asymptotic normality and provide explicit formulae for the limiting covariance matrix.

The remainder of this paper is organized as follows. In Section \ref{sec:meth}, we formalize the problems of the data challenge and detail the construction and properties of the estimators, including two ways for the construction of normal confidence intervals. These tools are then applied to the problems of the challenge in Section \ref{sec:application}, forming the core of this paper. 

\section{Methods}\label{sec:meth}

We consider the general task of estimating probabilities of rare spatio-temporal events in multivariate (precipitation) data. 

Let $(\Omega, \mathcal{A}, \mathbb{P})$ be a probability space. In this paper, all random variables are assumed to be defined on this space. Denote with $(\xi_t^{(\ell)})_{t}\subset \R^{d}, t=1,...,n;~ \ell=1,...,4$ the ($d=25$-dimensional) vector of daily precipitation at day $t$ in the simulation run $\ell$. The time series $(\xi_t^{(\ell)})_{t}$ are regarded as independent observations of a stochastic process $(\xi_t)_{t}$ with observation length $n=365\cdot165$.  
The $i^\text{th}$ coordinate of $\xi_t$ shall be $\xi_{t, i}$ and the $i^\text{th}$ smallest entry of $\xi_t$ shall be $\xi_{t, i:d}$.
Each task of the EVA 2025 Data Challenge corresponds to estimating an expectation (target) of the form
\begin{align*}
    \mathrm{T}_A := \mathbb{E} \left[ \sum_{t=1}^{n} \mathds{1}_{A_t} \right] = n \cdot \P(A_t),
\end{align*}
where $A_t= A_t((\xi_s)_{s\in\mathscr{S}(t)})$ is a generic spatio-temporal event describing the properties of $\xi_s$ in a temporal neighborhood $\mathscr{S}(t)$ of time $t$. The equality $\mathrm{T}_A=n \cdot \P(A_t)$ holds, of course, only under stationarity over $t$, which we want to assume throughout.

\begin{example}[Challenge events]\label{ex:events}
    We consider the following events: for a threshold $q>0$, put
    \begin{align}
        A_t^\mathrm{(P1)}(q) &:=\big\{\textstyle{\sum_{i=1}^d \xi_{t, i}} >q\big\} & 
        A_t^\mathrm{(C1)}(q) &:=\big\{\xi_{t,1:d} >q\big\} \notag\\\label{eq:the_events}
        A_t^\mathrm{(P2)}(q) &:=\big\{\xi_{t,3:5} >q\big\} & 
        A_t^\mathrm{(C2)}(q) &:=\big\{\xi_{t,6:d} >q\big\} \\
        A_t^\mathrm{(P3)}(q) &:=\big\{\xi_{t-1,3:5} \leq q,\xi_{t,3:5} \wedge\xi_{t+1,3:5} >q\big\} & 
        A_t^\mathrm{(C3)}(q) &:=\big\{\xi_{t-1,3:d} \leq q,\xi_{t,3:d} \wedge\xi_{t+1,3:d} >q\big\} \notag
    \end{align}
    The events (P1) -- (P3) and (C1) -- (C3) denote the targets of the preliminary and competition challenge, respectively.
\end{example}

Our subsequent approach is threefold. First, for a given $A_t$, we show how to construct univariate random variables $X_t$ with $\P(A_t(q))=\P(X_t>q)$, phrasing the targets as one-dimensional threshold exceedances. Next, we employ a peaks-over-threshold approach to model high threshold exceedances via a generalized Pareto distribution. Finally, we construct a $L_2$ Minimum Distance Estimator for $\P(A_t)$, whose asymptotic properties we leverage for inference. The remainder of this section formalizes the described steps.

\subsection{From events to univariate GPDs}
All events $A(q)$ of the data challenge have in common that they may be expressed as a function of a univariate parameter $q\geq 0$, to be called \textit{threshold}. Further, a common property is $\P(A(q))\to 0$ as $q\to\infty$. Contrary, the limit $q\to 0$ may differ between the events: for P1, P2, C1 and C2, we have $\P(A(q))\to 1$ as $q\to 0$ and, for P3 and C3, $\P(A(q))\to p\in[0,1]$. The qualitative difference of possible functions $q\mapsto \P(A(q))$ is depicted in the left-hand side of Figure~\ref{fig:construct}. During the data challenge, we are interested in modeling $\P(A(q))$ for large values of $q$, based on the observed counts $\#\{t: A_t(q)\}$. Assume that 
$ \P(A(q))\le \P(A(q'))$ for sufficiently large $q'<q$, the counts $\#\{t: A_t(q)\}$ are monotonically decreasing in $q$ (at sufficiently large $q$); which is not restrictive for the challenge events. Thus, it appears natural to interpret the rescaled counts as the empirical survival function of a  random variable. 
We want to construct such a random variable next and motivate a distributional assumption for it.

Recall that $q\mapsto \P(A(q))$ maps to $[0,1]$ with $\P(A(q))\downarrow 0$ as $q\to\infty$ for $q$ sufficiently large. Thus, we put
\begin{align}\label{eq:construct_F}
    && F(q)= 1-\sup_{w\ge q}\P(A_t(w)), && q\geq 0
\end{align}
with $F(q)=1-\P(A_t(q))$ for $q$ sufficiently large. See the right-hand side of Figure \ref{fig:construct} for (qualitative) examples of $1-F$. Note that $F$ is a cumulative distribution function on $[0,\infty)$ with a point mass of $1-\sup_{w\ge 0}\P(A_t(w))$ at $q=0$.\footnote{such a distribution, also known as a zero-inflated distribution, may occur, for instance, when modeling daily precipitation sums} Now introduce iid random variables $X_t\sim F, t=1,\dots,n$, and denote with $t_1,\dots, t_k$ the time points of exceedances, $\{t_1,\dots, t_k\}:=\{t:X_t>u\}$. The exceedances shall be $Y_j:=X_{t_j}-u,j=1,\dots,k$. By construction, we have for $x:=q-u>0$,
\begin{align*}
    \P(Y_j>x)=\frac{\P(X_t>q)}{\P(X_t>u)}=\frac{\P(A(q))}{\P(X_t>u)}.
\end{align*}
This demonstrates how to express the target probability $\P(A(q))$ via a univariate random variable $Y$. The term $\P(X_t>u)$ may simply be approximated via its empirical probability; the term $\P(Y_j>x)$ is a typical threshold exceedance probability. Threshold exceedances are usually modeled by a GPD, as the following section elaborates.

\begin{figure}
    \centering
    \includegraphics[width=0.75\linewidth]{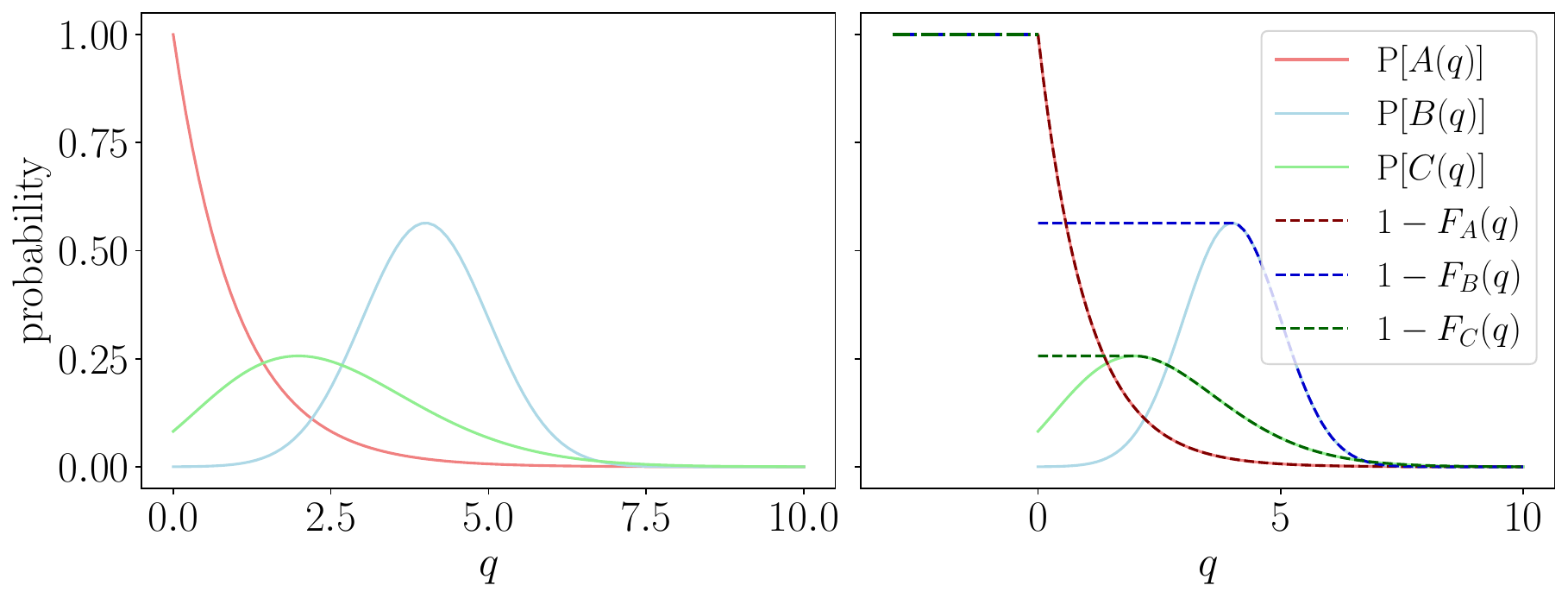}
    \caption{Left: qualitatively possible curves of $q\mapsto \P(E(q))$ for different generic events $E\in\{A,B,C\}$. Event $A$ qualitatively corresponds to P1, P2, C1 and C2; $B$ (if $p=0$) and $C$ (if $p>0$) to P3 and C3. Right: event probability curves $q\mapsto \P(E(q))$ together with their corresponding survival function $1-F_E(q)=\sup_{w\ge q}\P(A_t(w))$ as constructed in Equation \eqref{eq:construct_F}.}
    \label{fig:construct}
\end{figure}

\subsection{Peaks-over-Threshold Framework}
When analyzing exceedances of a random variable $X$ above a high threshold $u$, univariate extreme value theory motivates the use of Generalized Pareto Distribution (GPD) as a limiting model for excesses. The GPD cumulative distribution function is defined for any $x \ge 0$ by
\[
    F_\theta(x) = 1 - S_\theta(x) = 1- \left( 1+\gamma \frac{x}{\sigma} \right)_+^{-1/\gamma}, \quad \theta = (\gamma, \sigma) \in \Theta \subset \mathbb{R} \times (0,\infty),
\]
where $(z)_+ = \max(z,0)$, $\gamma$ is a shape parameter and $\sigma$ is a positive scale parameter.

We now motivate why and how to model the multivariate events $A_t$ of Example \ref{ex:events} with a GPD: let $X_1, X_2, \dots, X_n \overset{\mathrm{iid}}\sim F$ denote the random univariate variables constructed as described in the previous section. For a high threshold $u$, and excess level $x \ge 0$, define the conditional distribution of excesses:
\begin{align*}
    F_u(x) := \mathbb{P}(X_1 - u \le x \mid X_1 > u), \quad x \ge 0.
\end{align*}
The Pickands-Balkema-de Haan theorem \cite{Pickands1975,Balkema1974}
gives conditions, so that for sufficiently large $u$, $F_u(x) \approx F_\theta(x)$
where $F_\theta$ is the GPD with shape parameter $\gamma \in \mathbb{R}$ and $\sigma > 0$. In particular, exceedance probabilities admit the classical approximation
\begin{align}
    \P(X_1 > u + x) \approx \P(X_1 > u) \cdot S_{\theta}(x) \label{eq:splititup}
\end{align}
The exceedance probability $\mathbb{P}( X_1 > u )$ can be estimated empirically by $\hat {\P}(X_1 > u)=\frac1n\sum_{t=1}^n \indic(X_t>u)$, so it remains to focus on estimating the GPD parameters $\theta=(\gamma,\sigma)$.
\begin{example}
    To illustrate how multivariate tail behavior leads to a univariate POT limit, we begin with a concrete example. Let $\xi_t \sim Y$ for any $t = 1,\dots,n$ where $Y$ is a $d$-dimensional real random vector. Suppose that the $d$-dimensional random vector $Y = (Y_1,\dots,Y_d)$ follows max-linear model, i.e., $Y_j = \max_{a=1}^K A_{ja} Z_a$, $j = 1,\dots,d$ with nonnegative coefficients $A_{ja}$ and independent $\alpha$-Fréchet $Z_a$, $a = 1,\dots,K$ with $K \in \mathbb{N}_{\ge 1}$. Consider the failure region $C(x) = \{ y : \max(y_1,\dots,y_n) > x \}$, which is standard in environmental applications such as wind speed or extreme rainfall. Then, classical multivariate extreme value theory yields the approximation for exceedance of $X = \max\{Y_1,\dots,Y_d\}$ for large $x$
\[
\mathbb{P}\left\{ X > x \right\} \approx\nu_{Y}(C(x)) = \frac{1}{x^\alpha} \sum_{a=1}^K \max_{j=1,\dots,d} A_{ja}^\alpha,
\]
where $\nu_Y$ is the so-called exponent measure of the max-linear random vector $\xi$, see, e.g., \cite[Equation (9)]{kiriliouk2022estimating}, \cite[Lemma 2.7]{boulin2025structured}. Consequently, the tail of $X$ is regularly varying with index $\alpha$, so exceedances over a high threshold follow a GPD limit with shape parameter $\gamma = 1/\alpha > 0$ and scale parameter $\sigma = \frac{1}{\alpha} \left( \sum_{a=1}^K \max_{j=1,\dots,d} A_{ja}^\alpha \right)^{1/\alpha}$. Similar representation can be done for failure region such as $\{\mathbf{y} : \min(y_1,\dots,y_d) > x\}$, $\{\mathbf{y} ; \sum_{j=1}^d v_jy_j > x\}$ where $v_j > 0$, $j=1,\dots,d$ or, more generally, for any homogeneous transformation $h : [0,\infty)^d \rightarrow \mathbb{R}$ of a regularly varying random vector $\mathbf{Y}$ such that $\nu_{\mathbf{Y}} \circ h^{-1}$ is not the null-measure using \cite[Proposition 2.1.12]{kulik2020heavy}.
This illustrates how multivariate regular variation naturally leads to a univariate GPD limit for exceedances of suitable transformations.
\end{example}

\subsection{Minimum Distance Estimation (MDE)}\label{subs:mde}
Recall $X_1, X_2, \dots, X_n \overset{\mathrm{iid}}\sim F$, fix a high threshold $u = u_n \in \mathbb{R}$, and define the (random) number of exceedances as $K=K_n := \sum_{i=1}^n \indic_{\{X_i > u\}}$. Denote by $\{i_1, \dots, i_K\} \subset \{1, \dots, n\}$ the indices of those observations that exceed the threshold. For each $j = 1, \dots, K$, we define the excesses above $u$ by $Y_j := X_{i_j} - u.$ Thus, put $Y_1, \dots, Y_K \in (0, \infty)$ to denote the observed exceedances, which form the basis for GPD modeling. Define the empirical survival function to be
\begin{align*}
    \hat S_K(x) := \frac1K \sum_{j=1}^K \indic(Y_j>x).
\end{align*}
To construct the MDE, we will minimize the distance functional
\begin{align*}
    d(F,G):=\int_0^\infty [F(x)-G(x)]^2 \diff x,
\end{align*}
for any two functions $F,G \in L_2$. If $\gamma<2$, then $S_\theta \in L_2$ and we can we define
\begin{align*}
    && \hat\theta_K^\mathrm{MDE}:= \arg\min_{\theta\in\Theta} J_K(\theta), && J_K(\theta) = d( \hat S_K, S_\theta).
\end{align*}
By the plug-in principle, we obtain an estimator for the second term of the desired exceedance probability \eqref{eq:splititup}, $\hat S_K^\mathrm{MDE}(x)= S_{\hat\theta_K^\mathrm{MDE}}(x)$. To avoid technical complications with dependence within exceedances, only approximate validity of the GPD law, and $K$ being random, we will derive asymptotic properties for 
\begin{align*}
    \hat\theta_k^\mathrm{MDE}:= \arg\min_{\theta\in\Theta} J_k(\theta), && J_k(\theta) = d( \hat S_k, S_\theta), &&\hat S_k(x) := \frac1k \sum_{j=1}^k \indic(Z_j>x),
\end{align*}
with $Z_1,\dots, Z_k\sim \GPD(\theta_0)$ iid and $k\in\N$ deterministic. The sample $(Z_j)_j$ shall be thought of approximating $(Y_j)_j$.
The results below are stated under quite restrictive conditions, but sufficient to apply the estimator to the challenge. Most importantly, see Section \ref{sec:application} to verify $\gamma\in[0,1]$. In particular, they allow us to formulate the MDE as a $Z$-estimator.
\begin{condition}\label{cond:1}
    The random variables $Z_1,Z_2,\dots$ are iid and follow exactly a $\mathrm{GPD}(\theta_0)$ distribution, with $\theta_0\in\mathrm{Interior}(\Theta)$, where $\Theta\subset (0,1)\times (0,\infty)$ is a compact and convex subset.
\end{condition}
\begin{lemma}[MDE is a $Z$-estimator]\label{lem:write_as_z}
    Let $(z_1, \dots, z_k) \in (0,\infty)^k$ and $\Theta\subset (0,1)\times (0,\infty)$ compact. Any local minimizer of the optimization problem
    \begin{align*}
        \arg\min_{\theta\in\Theta} d\Big(\frac1k \sum_{j=1}^k \indic(z_j>x), S_\theta\Big)
    \end{align*}
    satisfies
    \[
        \Psi_k(\theta):=\frac{1}{k}\sum_{j=1}^k \psi(z_j, \theta) = 0, 
    \]
    where, for $x\in(0,\infty)$, $\psi(x,\theta) = \begin{pmatrix}
        \psi_\gamma(x,\theta) & \psi_\sigma(x,\theta)
    \end{pmatrix}^\top$ and
    \begin{align*}
        \psi_\gamma(x,\theta) &= \frac{\sigma }{2 (\gamma -2)^2}+(\gamma -1)^{-2} \gamma ^{-2}\Big[-\gamma ^2 \sigma+\Big(\frac{\sigma }{\sigma +\gamma  x}\Big)^{1/\gamma } \Big\{\gamma  (\gamma  \sigma +(2 \gamma -1) x)\\
        &\phantom{{}={}}\hspace{5cm}-(\gamma -1) (\sigma +\gamma  x) \log \left(\frac{\gamma  x}{\sigma }+1\right)\Big\} \Big] \\
        \psi_\sigma(x,\theta) &= -\frac{1}{2 (\gamma -2)}-(\gamma -1)^{-1} \sigma^{-1}\Big[(\sigma +x) \Big(\frac{\sigma }{\sigma +\gamma  x}\Big)^{1/\gamma }-\sigma  \Big].
    \end{align*}
\end{lemma}
By standard $Z$-estimator theory, we conclude two major results. 
\begin{theorem}[Consistency]\label{cor:consist_theta}
    Assume Condition \ref{cond:1}. Then any sequence $\hat\theta_k^\mathrm{MDE}=\hat\theta_k^\mathrm{MDE}(Z_1,\dots,Z_k)$ with $J_k(\hat\theta_k^\mathrm{MDE})\leq J_k(\theta_0)+o_\P(1)$
    , as $k\to\infty$,
    \begin{align*}
        \hat\theta_k^\mathrm{MDE} \xlongrightarrow{\P} \theta_0.
    \end{align*}
    Furthermore, for $x\ge 0$,
    \begin{align*}
        &&\hat S_k^\mathrm{MDE}(x)= S_{\hat\theta_k^\mathrm{MDE}}(x)\xlongrightarrow{\P} S_{ \theta_0}(x):=S_0(x).
    \end{align*}
\end{theorem}

\begin{theorem}[Asymptotic Distribution]\label{thm:normal_theta}
    Assume Condition \ref{cond:1}. Then any consistent estimator sequence $\hat\theta_k^\mathrm{MDE}$ with $\Psi_k(\hat\theta_k^\mathrm{MDE})=o_\P(1/\sqrt{k})$, as $k\to\infty$
        \begin{align*}
            \sqrt{k}\big(\hat\theta_k^\mathrm{MDE}-\theta_0\big) \xlongrightarrow{\mathcal D} Z\sim \mathcal N_2(0,\Sigma_{\theta_0}),
        \end{align*}
        with $\Sigma_{\theta_0}$ being explicitly given in the proof and visualized in Figure \ref{fig:cov_entries}. Second, for $x\geq 0$,
        \begin{align*}
            &&\sqrt{k}\big(\hat S_\mathrm{MDE}(x)-S_{0}(x)\big) \xlongrightarrow{\mathcal D} \mathbb G(x):= \nabla_\theta S_0(x)^\top Z.
        \end{align*}
        where $\mathbb G(x)$ is a mean-zero Gaussian Process with covariance kernel 
        \begin{align*}
            \Cov(\mathbb G(x),\mathbb G(x'))=\nabla_\theta S_0(x)^\top \Sigma_{\theta_0} \nabla_\theta S_0(x').
        \end{align*}
        In particular, $\varsigma^2_\theta(x):=\Cov(\mathbb G(x),\mathbb G(x))$ satisfies
        \begin{align*}
            &\phantom{=}(3 (3-2 \gamma )^3 (\gamma -6)^2 (\gamma -3)^2 \gamma ^4 (\sigma +\gamma  x)^2)\cdot\varsigma^2_\theta(x) \\
            &= 4 \Big(\frac{\gamma  x}{\sigma }+1\Big)^{-2/\gamma } \bigg[4 \gamma ^2 (\gamma +2)^2 \Big(\gamma  \Big(\gamma  \Big(30 \gamma ^2-266 \gamma +(\gamma -2)^2 (\sigma +\gamma  x) \log \Big(\frac{\gamma  x}{\sigma }+1\Big)+857\Big)\\
            &\phantom{{}={}}-1208\Big)+639\Big) x^2\cdot \bigg\{-\Big(\Big(\gamma  \Big(\gamma  \Big(2 \gamma  \Big(4 \gamma ^2-58 \gamma +243\Big)-683\Big)+452\Big)-639\Big) (\gamma -2)^2\\
            &\phantom{{}={}}\cdot(\sigma +\gamma  x) \log \Big(\frac{\gamma  x}{\sigma }+1\Big)\Big)
            -4 \gamma  (\gamma +2) \Big(2 \gamma  \Big(\gamma  \Big(8 \gamma ^2-62 \gamma +223\Big)-415\Big)+639\Big) x\bigg\}\bigg]
        \end{align*}
\end{theorem}

\begin{figure}
    \centering
    \includegraphics[width=\linewidth]{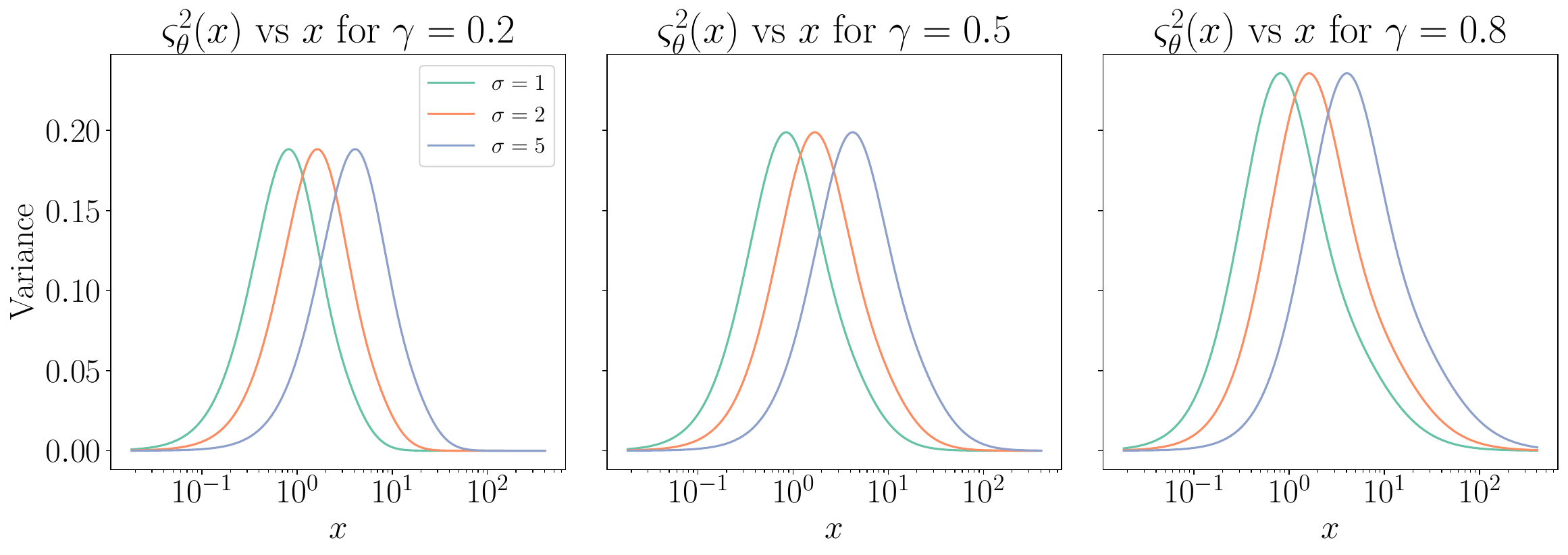}
    \caption{A depiction of the variance functions $x\mapsto\varsigma^2_\theta(x)$ for different choices of $\gamma$ and $\sigma$.}
    \label{fig:sigmax}
\end{figure}

The relevance in Theorem \ref{thm:normal_theta} lies in the explicit formula for $\varsigma_\theta^2(x)$. Utilizing the variance estimator $\hat\varsigma_k^2(x):=\varsigma_{\hat\theta_k^\mathrm{MDE}}^2(x)$, it allows to construct an asymptotically valid level-$\alpha$ confidence interval for $S_\theta(x)$ via
\begin{align}\label{eq:CI}
    \mathcal{C}_\theta(x):= \big[\hat S_\mathrm{MDE}(x) \pm \sqrt{k}\cdot q_{1-\alpha/2}\cdot\hat\varsigma_k(x)\big],
\end{align}
where $q_{1-\alpha/2}$ denotes the ${1-\alpha/2}$-quantile of the standard normal distribution.

\section{Application to the Data Challenge} \label{sec:application}
We now apply the methodology introduced in Section~\ref{sec:meth} to the three three preliminary (P1) -- (P3) and three competition targets (C1) -- (C3) from the EVA 2025 Data Challenge. Recall that each target corresponds to an expected number of spatio-temporal exceedances in a 25-dimensional grid over 165 years of daily observations. In particular, we are interested in the estimation of
\begin{align*}
    A_t^\mathrm{(P1)}(85), && A_t^\mathrm{(P2)}(4.3), && A_t^\mathrm{(P3)}(2.5), &&
    A_t^\mathrm{(C1)}(1.7), && A_t^\mathrm{(C2)}(5.7), && A_t^\mathrm{(C3)}(5),
\end{align*}
with $A_t^\mathrm{(T)}(q)$ as in Equation \eqref{eq:the_events}. For each target, we are provided four model runs $(\xi_t^{(\ell)})_t$, assumingly independent copies of $(\xi_t)_t$. Every copy leads to events $A_t^{\mathrm{(T)},\ell}(q)$, with equal probability $\P(A_t^{\mathrm{(T)},\ell}(q))=\P(A_t^{\mathrm{(T)},r}(q))$ for all $\ell,r=1,2,3,4$. 

Following the construction of Section \ref{subs:mde}, 
we denote $F^\mathrm{(T)}(q)=1-\sup_{w\geq q}\P(A_t^{\mathrm{(T)}}(w))$ and put $X_t^{(\mathrm{T}),\ell} \sim F^\mathrm{(T)}$. Further, fix a threshold $u\in(0,\infty)$, to be chosen later. We put $k_\ell:=|\{t\in[n]:X_t^{(\mathrm{T}),\ell}>u\}|$ and denote the indices (time points) of exceedances in copy $\ell$ with $\{i_1,\dots,i_{k_\ell}\}=\{t\in[n]:X_t^{(\mathrm{T}),\ell}>u\}$. The exceedances themselves shall be the random variables $Y_j^{(\mathrm{T}),\ell}:=X_{i_j}^{(\mathrm{T}),\ell}-u$ for $j=1,\dots,k_\ell$. Their survival functions are given by
\begin{align*}
    \hat S_k^{(\mathrm{T}),\ell}(x) := \frac1{k_\ell} \sum_{j=1}^{k_\ell} \indic(Y_j^{(\mathrm{T}),\ell}>x).
\end{align*}
With $Y_j^{(\mathrm{T}),\ell}$ being independent over $\ell$, we can obtain a more refined estimator for $S(x)$ by averaging over the four runs, leading to
\begin{align*}
    \hat S_k^{(\mathrm{T})}(x) := \frac1{k} \sum_{\ell=1}^4\sum_{j=1}^{k_\ell} \indic(Y_j^{(\mathrm{T}),\ell}>x),
\end{align*}
where $k=k_1+k_2+k_3+k_4$. The graphs of $\hat S_k^{(\mathrm{T})}(x)$ are depicted together with their MDE fits in \ref{fig:fits_c} for targets C1 -- C3, the preliminary targest P1 -- P3 are shown in the appendix, Figure \ref{fig:fits_p}. In contrast to the theory in Section \ref{subs:mde}, we fit a three-parametric GPD, $\vartheta=(\gamma,\mu,\sigma)$ to the survival function
\begin{align}
    \hat{S}_\vartheta(x):=\left(1 + \gamma\frac{ x-\mu}{\sigma} \right)^{-1/\gamma}_+, \quad x \ge \mu, \quad \vartheta = (\gamma,\mu,\sigma)\in\Xi\subset \R^2\times (0,\infty),
\end{align}
allowing for a more flexible model. 

One of the practical advantages of the proposed method is its simplicity, which goes hand in hand with its computational speed: a few seconds on a standard laptop.

\begin{figure}
    \centering
    \includegraphics[width=\linewidth]{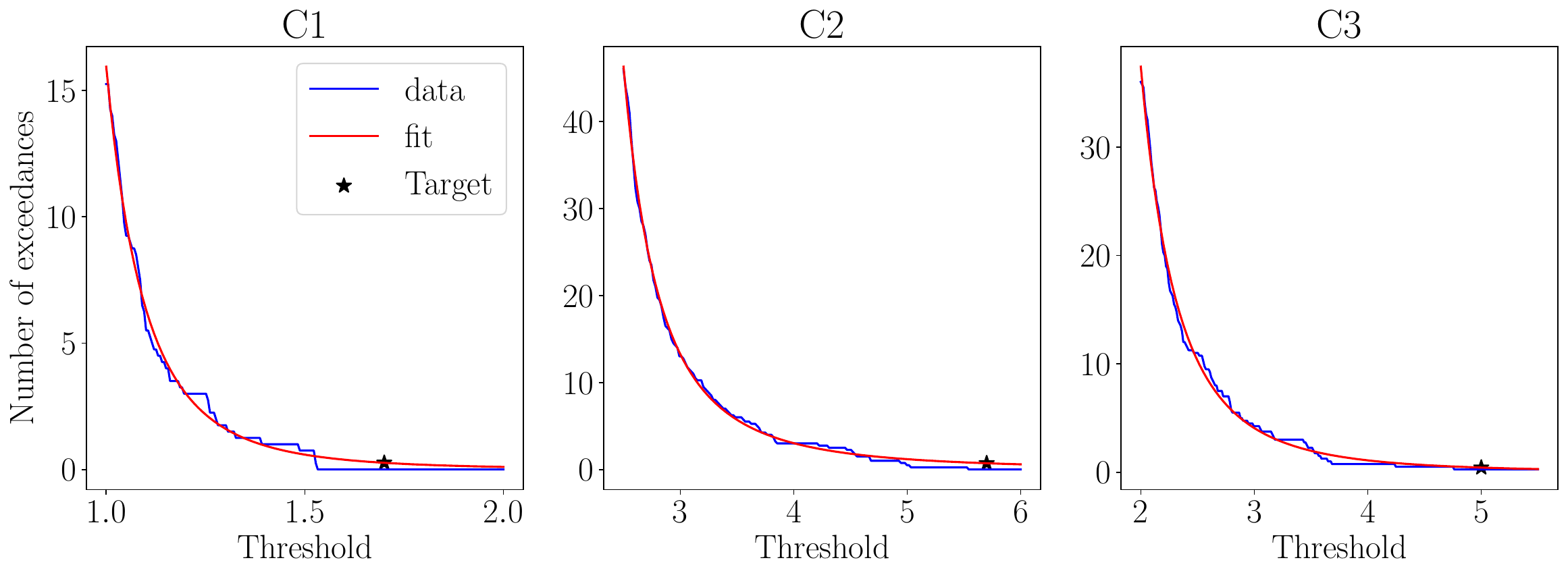}
    \caption{Average of empirical targets $\hat{\mathrm{T}}_A(q)= \sum_{t=1}^{n} \mathds{1}(A_t(q))$ (blue) and corresponding MDE fits $\hat{\mathrm{T}}_A^\mathrm{MDE}(q)$ (red), targets C1 -- C3, defined in Example~\ref{ex:events}. The black $\star$ denotes the target of the respective challenge.}
    \label{fig:fits_c}
\end{figure}

\subsection{Residual-based estimation of $\varsigma$}\label{sec:residual-estim}
At the time of the data challenge, Theorem \ref{thm:normal_theta} was not available, so that the estimator $\hat\varsigma_k$ as a function of $\hat\theta$, as used in Equation \eqref{eq:CI}, could not be utilized. We therefore employed a residual-based approach for the estimation of $\varsigma_\theta^2(x)$, which shall be shortly discussed. Simulation results strongly suggest to prefer the plug-in estimator $\hat\varsigma^2_k(x)$.

Consider the absolute residuals between the empirical survival function and the MDE fit
\begin{align*}
    \hat r_k(x) := |\hat S_k(x) - \hat S_k^\mathrm{MDE}(x)|.
\end{align*}
Inspired by heteroscedastic regression, e.g. \cite{Glejser1969}, 
the residuals $\hat r_k(x)$ were treated as approximations to the standard deviation $\varsigma_k(x)$. We therefore fit a parametric model to the residuals, see Figure \ref{fig:resids_c}, and choose, for the sake of simplicity, the one-parametric model
\begin{align}\label{eq:simple_sig}
    \varsigma_\phi(x):= \phi\cdot\sqrt{\hat S_k^\mathrm{MDE}(x)}.
\end{align}

The parameter $\phi$ is finally determined via MDE,
\begin{align*}
    \hat\phi_k:= \arg\min_{\phi>0} \int_0^\infty [\hat r_k(x)- \varsigma_\phi(x)]^2 \diff x.
\end{align*}
\begin{figure}
    \centering
    \includegraphics[width=\linewidth]{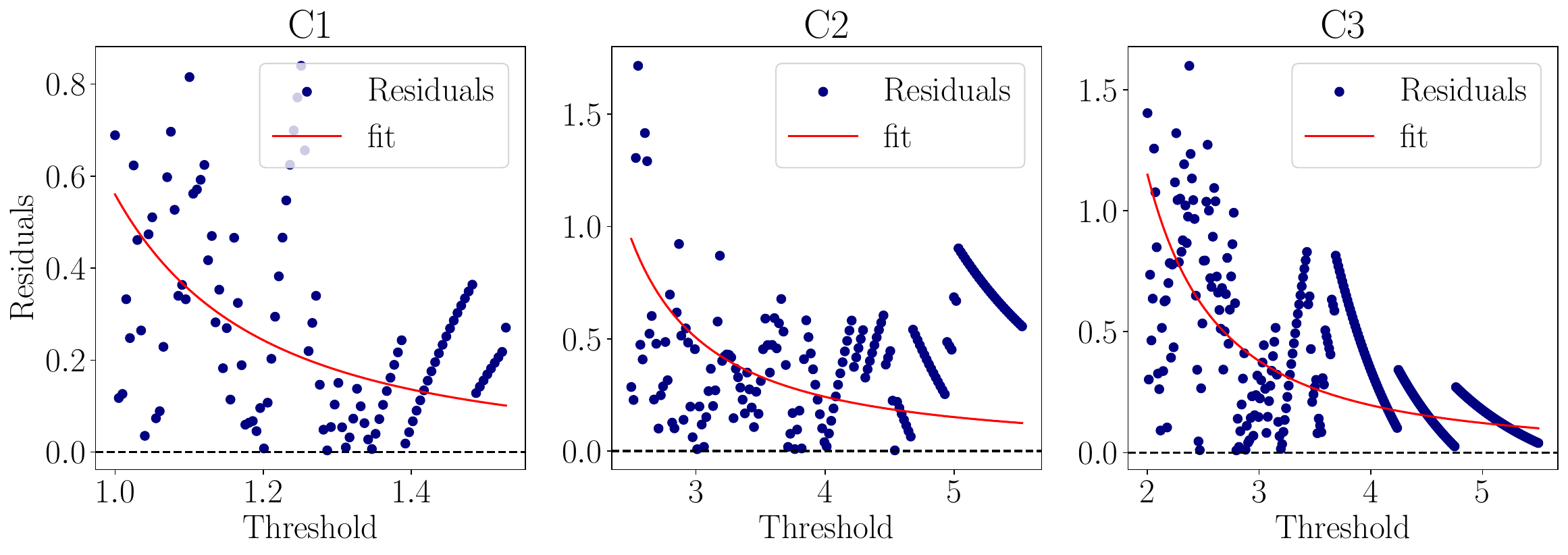}
    \caption{Residuals and fitted $\tilde\varsigma_k^2(x)$ for target C1 -- C3. The other residual fits have been postponed to Appendix \ref{sec:additional}, Figure \ref{fig:resids_p}.}
    \label{fig:resids_c}
\end{figure}

\subsection{On the choice of threshold}

All previous sections assumed a known threshold $u=u_n$. As it has to be chosen in practice, we investigate how the estimated value $\tilde S_k^\mathrm{MDE}(x)$ changes with $u$, the results are depicted in Figure \ref{fig:thchoice_c} and \ref{fig:thchoice_p}. It is particularly notable that each of the six targets shows a similar behavior with varied threshold $u$: 
\begin{itemize}
    \item For very low thresholds, the GPD is fitted to the whole bulk of data, not only to extreme exceedances. Consquently, the GPD model does not describe the data well and we see large estimated values.
    \item For intermediate thresholds, a close-to-zero prediction is visible in all six targets. This may indicate that the intermediate parts of these distributions are not heavy-tailed enough to capture the actual behaviour of extreme observations adequately.
    \item For high thresholds, a stabilizing region may be identified, over which the estimated values do not change too much. A zoomed-in version of the ``stabilized regions'' is shown in Figures \ref{fig:thchoice_c_zoom} and \ref{fig:thchoice_p_zoom}. 
\end{itemize}

As it still remains open how to choose the threshold within the stabilized region adequately, we circumvented this issue in two different ways: for the preliminary challenge, we chose a threshold $u$, such that the resulting fit appears to describe the empirical survival probabilities adequately, based on visual diagnostics. For the competition targets, we identified a stabilized region $[u_1,u_2]$ via visual diagnostics and averaged the predictions over this interval. The choices of thresholds are summarized in Table \ref{tab:summary}, together with the point estimates and confidence intervals. Note that the parameter values of the competition targets are only exemplary parameters -- they change over the interval the estimators are being averaged over.
\begin{figure}
    \centering
    \includegraphics[width=\linewidth]{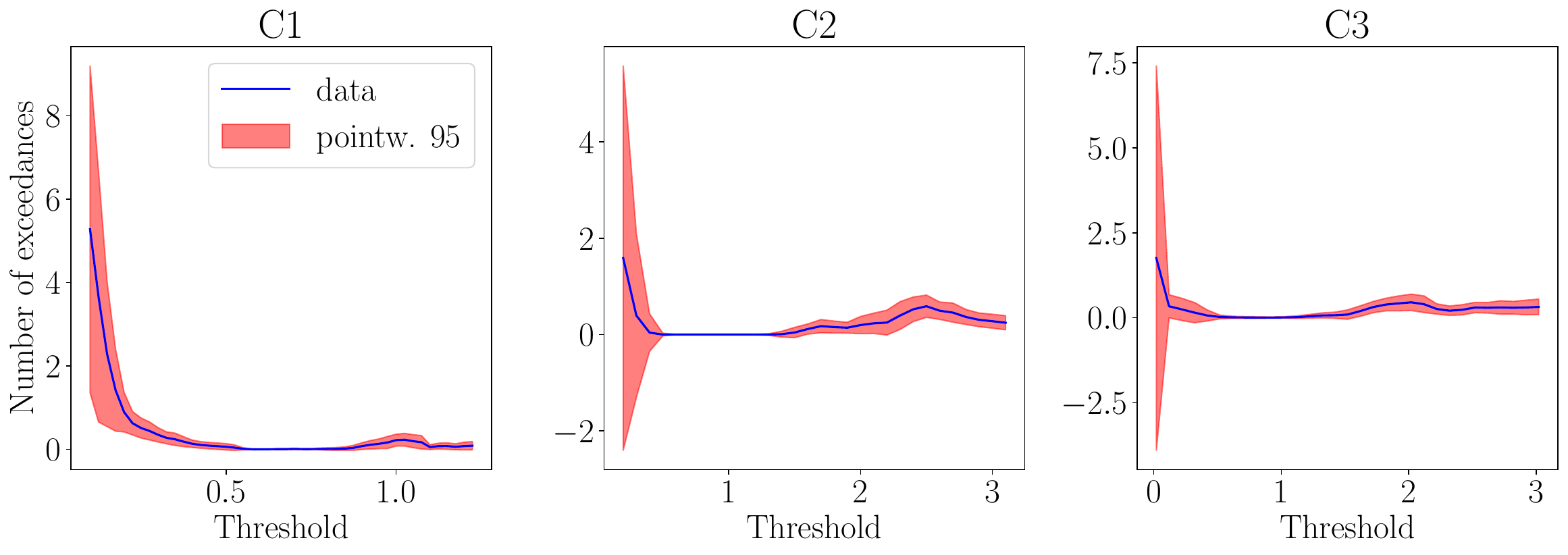}
    \caption{Development of the estimator $\hat{\mathrm{T}}_A^\mathrm{MDE}(q)$ over the threshold $u$ for target C1 -- C3. The extreme threshold $q=x+u$ is the target threshold and remains constant.}
    \label{fig:thchoice_c}
\end{figure}
\begin{figure}
    \centering
    \includegraphics[width=\linewidth]{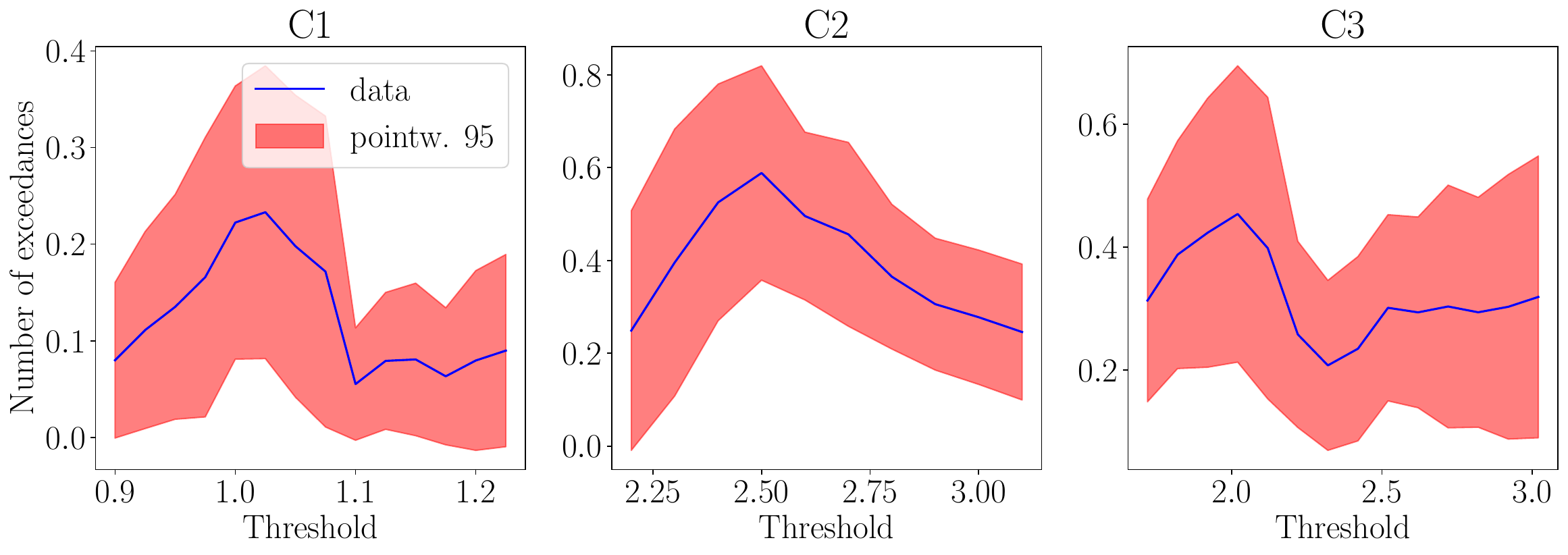}
    \caption{Zoom-in: development of the estimator $\hat{\mathrm{T}}_A^\mathrm{MDE}(q)$ over the threshold $u$ for target C1 -- C3. The extreme threshold $q=x+u$ is the target threshold and remains constant.}
    \label{fig:thchoice_c_zoom}
\end{figure}
\begin{remark}[On the choice of threshold]
    There exists a vast literature reviewing tools for choosing the threshold sequence $u_n$ in practice, motivated either by theory \cite{fougeres2015bias,scarrott2012review, northrop2017cross} 
    or by applications \cite{pan2022peaks, solari2017peaks, duran2022mixed}, among many others. The discussion about which method should be preferred remains ongoing, and no clear consensus has emerged, as recently highlighted in a contributed abstract at the EVA 2025 conference \cite{Belzile2025EVA}. We do not aim to contribute further to this debate and simply advise practitioners to adopt their preferred threshold‐selection method.
\end{remark}
\begin{table}[h!]
\centering
\caption{Summary of estimation results for targets (P1) -- (P3) and (C1) -- (C3). The target level $q=x+u$ is split into an excess $x$ and a threshold $u$, with choices of $u$ in third column. Consecutive columns contain the $L_2$-MDE estimated parameters of the $\GPD(\gamma,\mu,\sigma)$ distribution, the point estimate for the target probability and a corresponding 95\% CI.}
\renewcommand{\arraystretch}{1.5}
\label{tab:summary}

\begin{tabular}{lccccc}
\toprule
Target & $q$ & $u$ [$(u_1,u_2)$] & $(\hat{\gamma}, \hat{\mu}, \hat{\sigma})$ & Point estimate & 95\% CI\\
\midrule
P1 & 85 &  43 & (0.168, 20.71, 1.615) & 0.3218 & (0.1336, 0.5101) \\
P2 & 4.3 &  2.1 & (0.199, 0.914, 0.070) & 0.4160 & (0.2970. 0.5351) \\
P3 & 2.5 &  1.2 & (0.072, -0.0755, 0.125) & 0.1978 & (0.1080, 0.2875) \\
C1 & 1.7 &  [42, 70] & (0.257, 0.671, 0.012) & 0.1281 & (0.0186, 0.2377) \\
C2 & 5.7 &  [2, 4] & (0.383, 1.738, 0.020) & 0.4405 & (0.2064, 0.6746) \\
C3 & 5.0 &  [0.9, 1.3] & (0.302, 1.065, 0.034) & 0.3144 & (0.1291, 0.4996) \\
\bottomrule
\end{tabular}
\end{table}

\section{Conclusion}\label{sec13}

In this work, we presented a very simple approach to estimating spatio-temporal exceedance probabilities in multivariate precipitation data, as posed by the EVA 2025 Data Challenge. By reducing each task to a univariate peaks-over-threshold problem, we demonstrate how to efficiently avoid the whole field of multivariate extreme value statistics -- even though all tasks targeted this particular field of study. 

Within the peaks-over-threshold problem, we investigate an $L_2$ estimator for the parameters of a generalized Pareto distribution. Despite estimation via minimum distance is not novel at all, the particular $L_2$ distance under consideration has, to the best of our knowledge, not been studied theoretically for the GPD distribution so far. We derive asymptotic theory for the $L_2$-MDE by embedding it into the classical framework of $Z$-estimation. %

Further, we provide a detailed comparison of our $L_2$-MDE to the classical maximum likelihood estimator in the appendix. Monte Carlo simulations indicate superiority of MLE over $L_2$-MDE, thus we cautiously recommend using former in application settings.

Notwithstanding the simplicity of our approach, we managed to be competitive to stringent and complex modeling approaches from multivariate extreme value analysis, in the application to the EVA2025 data challenge.

Beyond the scope of this competition, our approach may be considered as an unconventional, but successful alternative to problems in EVT, its main idea being universally applicable to a large scope of different problems. Besides, the discussed properties of the $L_2$-MDE may shed new light on the usage of minimum distance estimators for peaks-over-threshold problems.

\subsection*{Acknowledgments}

The authors would like to express their thanks to Dan Cooley, Ben Shaby (Colorado State), Jennifer Wadsworth (Lancaster) and Emily Hector (Michigan) for organizing the data challenge of the 14$^\text{th}$ International Conference on Extreme Value Analysis 2025. 
Finally, valuable hints from and discussions with Axel Bücher are gratefully acknowledged.

\section*{Declarations}

\subsection*{Funding} This work has been supported by the integrated project ``Climate Change and Extreme Events  -- ClimXtreme Module B Statistics Phase II'' (project B3.3, grant number 01LP2323L) funded by the German Federal Ministry of Education and Research (BMBF). Further financial support by the Deutsche Forschungsgemeinschaft (DFG, German Research Foundation; Project-ID 520388526; TRR 391: Spatio-temporal Statistics for the Transition of Energy and Transport) is gratefully acknowledged.
Erik Haufs is grateful for support by the Studienstiftung des deutschen Volkes and by the Ruhr University Research School, funded by Germany’s Excellence Initiative [DFG GSC 98/3].
This work used resources of the Deutsches Klimarechenzentrum (DKRZ) granted by its Scientific Steering Committee (WLA) under project ID bb1152.

\subsection*{Conflict of interest}
The authors do not declare any conflicts of interest.

\subsection*{Data availability statement} The data supporting the findings of this study are from the EVA 2025 data challenge \cite{EVA2025DataChallenge}, and are available on request from the the organizers.

\subsection*{Author contribution} 
AB and EH participated together in the Data Challenge, conducting exploratory data analysis, simulations, and developed approaches. EH suggested the construction in Section 2. AB and EH wrote the manuscript.

\printbibliography
\appendix

\section{Proofs for Section \ref{sec:meth}}\label{sec:proofs}

\begin{proof}[Proof of Lemma \ref{lem:write_as_z}]
Recall that the gradient $\nabla J_k$ vanishes for any local minimizer. Thus, we will consecutively calculate the gradient.
Denote by $f_\theta(x)$ the density of a GPD($\theta)$, i.e.,
\[
    f_\theta(x) = \frac{1}{\sigma}\left(1+\frac{\gamma}{\sigma} x\right)^{-\frac{1}{\gamma} -1},
\]
For any $x\geq 0$
, the partial derivatives of $F_\theta$ are given by:
\begin{align*}
    &\frac{\partial}{\partial \gamma} F_\theta(x) = -\left(1+\frac{\gamma}{\sigma}x\right)^{-1/\gamma} \left(\frac{\log\left( 1+\frac{\gamma x}{\sigma}\right)}{\gamma^2} - \frac{\frac{1}{\gamma\sigma}x}{1+\frac{\gamma}{\sigma} x}\right) \\&\phantom{\frac{\partial}{\partial \gamma} F_\theta(x) {}}= -\frac{1}{\gamma^2\sigma} f_\theta(x) \left( (1+\gamma \frac x\sigma) (\log\left( 1+\gamma \frac x\sigma\right) -1) \right), \\
    &\frac{\partial}{\partial \sigma} F_\theta(x) = -\frac{x}{\sigma^2} \left(1+\frac{\gamma}{\sigma} x\right)^{-1/\gamma -1} = -\frac{1}{\sigma} x f_\theta(x).
\end{align*}
To compute the gradient of $J_k$, we want to swap the integral and the derivative sign. Since $\Theta$ is compact, there exist constants $0<\gamma_{\min} < \gamma < \gamma_{\max} < 1$ and $0 < \sigma_{\min} < \sigma < \sigma_{\max} < \infty$ such that the rectangle $(\gamma_{\min}, \gamma_{\max}) \times (\sigma_{\min}, \sigma_{\max})$ contains $(\gamma, \sigma)$ and lies within a neighborhood of $\Theta$. For the partial derivative with respect to $\gamma$, observe that for any $\theta \in \Theta$, 
\[
    \left| \frac{\partial}{\partial \gamma} F_\theta(x) \right| \le g_1(x),
\]
with,
\[
    g_1(x) = \left( 1 + \frac{\gamma_{\max}}{\sigma_{\max}} x \right)^{-1/\gamma_{\max}} 
\left( 
\frac{ \log\left( 1 + \frac{\gamma_{\max}}{\sigma_{\min}} x \right) }{ \gamma_{\min}^2 } 
+ 
\frac{ \frac{x}{\gamma_{\min} \sigma_{\min}} }{ 1 + \frac{\gamma_{\min}}{\sigma_{\max}} x } 
\right),
\]
where we use that $\gamma \mapsto (1+{\gamma x}/{\sigma})^{-1/\gamma}$ is non-decreasing 
and $\sigma > 0$. This can be seen by defining
\[
A(\gamma, \sigma, x) := \left(1 + \frac{\gamma x}{\sigma} \right)^{-1/\gamma}, \qquad
B(\gamma, \sigma, x) := \frac{\log\left( 1 + \frac{\gamma x}{\sigma} \right)}{\gamma^2}
- \frac{ \frac{x}{\gamma \sigma} }{ 1 + \frac{\gamma x}{\sigma} }.
\]
Then:
\[
\left| \frac{\partial}{\partial \gamma} F_\theta(x) \right| = |A(\gamma, \sigma, x) \cdot B(\gamma, \sigma, x)| \le A(\gamma, \sigma, x) \cdot |B(\gamma, \sigma, x)|.
\]
We now bound both terms. Due to $\gamma \mapsto (1+{\gamma x}/{\sigma})^{-1/\gamma}$ being non-decreasing, 
\[
    A(\gamma, \sigma, x) \le \left(1 + \frac{\gamma_{\max}}{\sigma_{\max}} x \right)^{-1/\gamma_{\max}}.
\]
Second, for $B(\gamma, \sigma, x)$, observe that
    \[
        B(\gamma, \sigma, x) = \frac{1}{\gamma^2} \phi\left(1+\frac{\gamma}{\sigma}x\right),
    \]
    where
    \[
\phi \colon (0,\infty) \to (0,\infty), \quad x \mapsto \phi(x) = \log(x) -1 + \frac{1}{x}.
    \]
    Since $\phi(\cdot)$ is increasing on $(1,\infty)$ we obtain by noticing that $(1+{\gamma x}/{\sigma}) \le (1+{\gamma_{\max} x}/{\sigma_{\min}})$:
    \[
    |B(\gamma, \sigma, x)| \le 
    \frac{ \log\left(1 + {\gamma_{\max}\cdot x}/{\sigma_{\min}} \right) }{ \gamma_{\min}^2 }
    + 
    \frac{ {x}/{(\gamma_{\min}\cdot \sigma_{\min})} }{ 1 + {\gamma_{\min}\cdot x/\sigma_{\max}} }.
    \]   
Combining the bounds, we conclude:
\[
\left| \frac{\partial}{\partial \gamma} F_\theta(x) \right| \le g_1(x),
\]
where:
\[
g_1(x) := \left(1 + \frac{\gamma_{\max}}{\sigma_{\max}} x \right)^{-1/\gamma_{\max}} 
\left( 
\frac{ \log\left( 1 + \frac{\gamma_{\max}}{\sigma_{\min}} x \right) }{ \gamma_{\min}^2 } 
+ 
\frac{ \frac{x}{\gamma_{\min} \sigma_{\min}} }{ 1 + \frac{\gamma_{\min}}{\sigma_{\max}} x } 
\right).
\]
We now show that $g_1(x)$ is integrable. Indeed
\[
    \int_0^\infty \left(1+\frac{\gamma_{\max}}{\sigma_{\max}} x \right)^{-1/\gamma_{\max}} \log\left( 1 + \frac{\gamma_{\max}}{\sigma_{\min}} x \right) \diff x = \frac{\sigma_{\max}\gamma_{\max}}{(\gamma_{\max}-1)^2} < \infty,
\]
where $0<\gamma_{\max} < 1$. Also,
\[
    \int_0^\infty \left( 1 + \frac{\gamma_{\max}}{\sigma_{\max}} x \right)^{-1/\gamma_{\max}} 
    \left(\frac{ \frac{x}{\gamma_{\min} \sigma_{\min}} }{ 1 + \frac{\gamma_{\min}}{\sigma_{\max}} x } 
    \right)\diff x < \infty.
\]
Same arguments applies for $\frac{\partial}{\partial \sigma} F_{\theta}$, therefore, by Lebesgue's theorem, we may interchange integration and differentiation and obtain $\hat\theta_k^\mathrm{MDE}$ as the (possibly not unique) root of
\begin{align*}
    \partial_\theta J_k(\theta)=\int_0^\infty \frac{\partial}{\partial \theta}[\hat F_k(x)-F_\theta(x)]^2~\mathrm{d}~\!x
\end{align*}
We conclude 
\begin{align}
    \psi(x,\theta) = \int_0^x \left\{ \frac{\partial}{\partial \theta} F_\theta(y) \right\} \diff y - \int_0^\infty \int_0^u \left\{ \frac{\partial}{\partial \theta} F_\theta(y) \right\} \diff y \diff F_\theta(u) \label{eq:psi_after_partint}
\end{align}

by partial integration in analogy to \cite{Clarke1994}. The assertion follows from evaluating the appearing integrals, see also Appendix \ref{app:mathematica}: the inner integrals are
\begin{align*}
    \int_0^x \left\{ \frac{\partial}{\partial \gamma} F_\theta(y) \right\} \diff y &= \frac{\left(\frac{\sigma }{\sigma +\gamma  x}\right)^{1/\gamma } \left(\gamma  (\gamma  \sigma +(2 \gamma -1) x)-(\gamma -1) (\sigma +\gamma  x) \log \left(\frac{\gamma  x}{\sigma }+1\right)\right)-\gamma ^2 \sigma }{(\gamma -1)^2 \gamma ^2}, \\
    \int_0^x \left\{ \frac{\partial}{\partial \sigma} F_\theta(y) \right\} \diff y &=-\frac{(\sigma +x) \left(\frac{\sigma }{\sigma +\gamma  x}\right)^{1/\gamma }-\sigma }{(\gamma -1) \sigma },
\end{align*}
whereas the outer integrals yield
\begin{align*}
    \int_0^\infty \int_0^u \left\{ \frac{\partial}{\partial \gamma} F_\theta(y) \right\} \diff y \diff F_\theta(u) &=-\frac{\sigma }{2 (\gamma -2)^2},\\
    \int_0^\infty \int_0^u \left\{ \frac{\partial}{\partial \sigma} F_\theta(y) \right\} \diff y \diff F_\theta(u) &=\frac{1}{2 (\gamma -2)}.
\end{align*}

\end{proof}

\begin{proof}[Proof of Theorem \ref{cor:consist_theta}]
    This is a direct consequence of \cite[Theorem 5.7]{Vaart1998}, it remains to verify its conditions. First, we need to show that 
    \begin{align*}
        \sup_{\theta\in\Theta} \|J_k(\theta)-J(\theta)\|\xlongrightarrow{\P}0.
    \end{align*}
    For this, we have by Cauchy-Schwartz and the triangle inequality
    \begin{align*}
         |J_k(\theta)-J(\theta)| &= \Big| \int_0^\infty (\hat S_k(x) - S_\theta(x))^2-(S_{\theta_0}(x)-S_\theta(x))^2\diff x\Big|\\
         &= \Big| \int_0^\infty (\hat S_k(x) - S_{\theta_0}(x))\big[\hat S_k(x)+S_{\theta_0}(x)-2S_\theta(x)\big]\diff x\Big|\\
         &=\Big\langle \hat S_k - S_{\theta_0},~\hat S_k+S_{\theta_0}-2S_\theta\Big\rangle_{\mathds L_2}^2 \\
         &\leq \big\|\hat S_k - S_{\theta_0}\big\|_{\mathds L_2}^2 \cdot\big\|\hat S_k-S_{\theta_0}+2[S_{\theta_0}-S_\theta]\big\|_{\mathds L_2}^2 \\
         &\leq \big\|\hat S_k - S_{\theta_0}\big\|_{\mathds L_2}^2 \cdot \Big\{\big\|\hat S_k-S_{\theta_0}\big\|_{\mathds L_2}+2\big\|S_{\theta_0}-S_\theta\big\|_{\mathds L_2}\Big\}^2.
    \end{align*}
    Consequently,
    \begin{align*}
        \sup_{\theta\in\Theta} |J_k(\theta)-J(\theta)| &\leq \big\|\hat S_k - S_{\theta_0}\big\|_{\mathds L_2}^2 \cdot \Big\{\big\|\hat S_k-S_{\theta_0}\big\|_{\mathds L_2}+2\sup_{\theta\in\Theta} \big\|S_{\theta_0}-S_\theta(x)\big\|_{\mathds L_2}\Big\}^2.
    \end{align*}
    For fixed $\theta_0\in\Theta$, we have
    \begin{align*}
        \sup_{\theta\in\Theta} \big\|S_{\theta_0}-S_\theta\big\|_{\mathds L_2}^2=\sup_{\theta\in\Theta}\int_0^\infty [S_{\theta_0}(x)-S_\theta(x)]^2 \diff x=:C_{\theta_0}<\infty,
    \end{align*}
    where the latter holds due to the compactness of $\Theta\subset (0,1)\times (0,\infty)$.
    Now, by Fubini,
    \begin{align*}
        \E\Big[\big\|\hat S_k - S_{\theta_0}\big\|_{\mathds L_2}^2\Big] &= \int_0^\infty \E\bigg[\Big[\frac{1}{k}\sum_{j=1}^k \indic(Z_j>x) - S_{\theta_0}(x)\Big]^2\bigg]\diff x 
        = \int_0^\infty \mathds V\mathrm{ar}\Big[\frac{1}{k}\sum_{j=1}^k \indic(Z_j>x) \Big]\diff x \\
        &= \frac1k \int_0^\infty S_{\theta_0}(x)[1-S_{\theta_0}(x)]\diff x 
        =:\frac{V_{\theta_0}}{k}<\infty,
    \end{align*}
    as $\hat S_k(x)$ is a rescaled $\mathrm{Bin}(k,S_{\theta_0}(x))$ random variable, pointwise. The latter $V_{\theta_0}<\infty$ again holds due to $\Theta\subset (0,1)\times (0,\infty)$. By Chebyshev's inequality,
    \begin{align*}
        \big\|\hat S_k - S_{\theta_0}\big\|_{\mathds L_2}^2=o_\P(1).
    \end{align*}
    We conclude 
    \begin{align*}
        \sup_{\theta\in\Theta} |J_k(\theta)-J(\theta)| &\leq o_\P(1)\cdot\big\{o_\P(1)+2\sqrt{C_{\theta_0}}\big\}^2=o_\P(1),
    \end{align*}
    so that the claimed convergence in probability is indeed fulfilled.
It remains to check the second condition
    \begin{align*}
        \inf_{\theta:~d(\theta,\theta_0)\ge \epsilon} \|J(\theta)\|>0=\|J(\theta_0)\|,
    \end{align*}
    which directly follows from the identifiability of the GPD.
    
    We conclude the second assertion by applying the continuous mapping theorem.
\end{proof}

\begin{proof}[Proof of Theorem \ref{thm:normal_theta}]
    We apply \cite[Theorem 5.21]{Vaart1998}. Having established Lemma \ref{lem:write_as_z}, the well-behaved derivatives of $\psi$ and the strong assumptions on $\Theta$ make most conditions easy to verify, see sufficient conditions in \cite[Theorem 5.41]{Vaart1998}. The only condition requiring major effort is the boundedness of the second-order partial derivatives of $\psi$ in a neighborhood around $\theta_0$, which we show later. It remains to calculate the limiting covariance matrix $\Sigma_{\theta_0}$. Following \cite{Vaart1998}, we have $\Sigma_{\theta} = U^{-1}_\theta V_\theta (U_\theta^\top)^{-1}$ with
    \begin{align*}
        U_\theta = \E\begin{bmatrix}
            \partial_\gamma\psi_\gamma(Y,\theta)& \partial_\gamma\psi_\sigma(Y,\theta) \\
            \partial_\sigma\psi_\gamma(Y,\theta)& \partial_\sigma\psi_\sigma(Y,\theta) 
        \end{bmatrix}
        ,&&
        V_\theta=\E\begin{bmatrix}
            \psi_\gamma(Y,\theta)^2 & \psi_\gamma(Y,\theta)\psi_\sigma(Y,\theta) \\
            \psi_\gamma(Y,\theta)\psi_\sigma(Y,\theta) & \psi_\sigma(Y,\theta)^2
        \end{bmatrix}.
    \end{align*}
    Evaluating these matrices, we find
    \begin{align*}
        U_\theta = \begin{pmatrix}
            \frac{(\gamma -6) \sigma }{2 (\gamma -2)^3 (\gamma +2)}&\frac{6-\gamma}{4 (\gamma -2)^2 (\gamma +2)}\\
            \frac{6-\gamma}{4 (\gamma -2)^2 (\gamma +2)}&\frac{1}{4 \sigma -\gamma ^2 \sigma }
        \end{pmatrix}.
    \end{align*}
    and
    \begin{align*}
        V_\theta &= \begin{pmatrix}
            \frac{\left(8 \gamma ^5-148 \gamma ^4+918 \gamma ^3-2587 \gamma ^2+3416 \gamma -1719\right) \sigma ^2}{12 (\gamma -3)^2 (\gamma -2)^4 (2 \gamma -3)^3} &  \frac{(\gamma  (\gamma  (-4 (\gamma -15) \gamma -285)+548)-369) \sigma }{12 (\gamma -2)^3 \left(2 \gamma ^2-9 \gamma +9\right)^2} \\
            \frac{(\gamma  (\gamma  (-4 (\gamma -15) \gamma -285)+548)-369) \sigma }{12 (\gamma -2)^3 \left(2 \gamma ^2-9 \gamma +9\right)^2} & \frac{\gamma  (2 \gamma -17)+29}{12 (\gamma -3) (\gamma -2)^2 (2 \gamma -3)}
        \end{pmatrix}.
    \end{align*}
    Concluding, {
    \begin{align}
        \Sigma_{\theta} &= U^{-1}_\theta V_\theta (U_\theta^\top)^{-1} \notag\\\label{eq:covvals}
        &= \left(
\begin{array}{cc}
 \frac{4 (\gamma -2)^4 \left(\gamma  \left(\gamma  \left(2 \gamma  \left(4 \gamma ^2-58 \gamma +243\right)-683\right)+452\right)-639\right)}{3 (\gamma -6)^2 (\gamma -3)^2 (2 \gamma -3)^3} & \frac{4 (\gamma -2)^2 \left(\gamma  \left(\gamma  \left(2 \gamma  \left(4 \gamma ^2-50 \gamma +207\right)-791\right)+778\right)-387\right) \sigma }{3 (\gamma -6) (\gamma -3)^2 (2 \gamma -3)^3} \\
 \frac{4 (\gamma -2)^2 \left(\gamma  \left(\gamma  \left(2 \gamma  \left(4 \gamma ^2-50 \gamma +207\right)-791\right)+778\right)-387\right) \sigma }{3 (\gamma -6) (\gamma -3)^2 (2 \gamma -3)^3} & \frac{4 \left(\gamma  \left(\gamma  \left(8 \gamma ^3-84 \gamma ^2+374 \gamma -843\right)+944\right)-423\right) \sigma ^2}{3 (\gamma -3)^2 (2 \gamma -3)^3} \\
\end{array}
\right)
    \end{align}}
Details of this calculation have been performed in \texttt{Mathematica} \cite{Mathematica}, which is elaborated in Appendix \ref{app:mathematica}. The entries of $\Sigma_\theta$ for selected values of $\theta$ are depicted in Figure \ref{fig:cov_entries}.

It remains to show the second assertion, to which we apply the delta method. As $\nabla_\theta S_0(x)\neq 0$ for fixed $x$, the convergence
\begin{align*}
    \sqrt{k}\big(\hat S_\mathrm{MDE}(x)-S_{0}(x)\big) \xlongrightarrow{\mathcal D} \mathbb G(x):= \nabla_\theta S_0(x)^\top Z
\end{align*}
obviously holds. The covariance structure of $\mathbb G(x)$ is a direct consequence. Finally, the explicit structure of $\varsigma_\theta^2(x)$ immediately follows from evaluating $\varsigma_\theta^2(x)= \nabla_\theta F_\theta(x)^\top \Sigma_\theta \nabla_\theta F_\theta(x)$. \\

\textit{Proof of the boundedness of the second-order partial derivatives.} Within the interior of $\Theta$, the second order derivatives are continuous in a neighborhood around $\theta_0$. It therefore suffices to show the integrability of the second-order partial derivatives at $\theta_0$. We have
\begin{align*}
    \int_0^\infty \frac{\partial^2}{\partial\sigma^2} \psi(x,\theta)\diff F_\theta(x) &=-\frac{1}{(\gamma ^2-4) \sigma }\\
    \int_0^\infty \frac{\partial^2}{\partial\sigma\partial\gamma} \psi(x,\theta)\diff F_\theta(x) &=-\frac{ G_{3,3}^{3,2}\left(1\left|
\begin{array}{c}
 -2,-\frac{2}{\gamma },-1 \\
 -2,-2,0 \\
\end{array}
\right.\right)}{ (\gamma -1) \gamma ^4 \Gamma (1+2/\gamma)}-\frac{\left(\gamma ^3-6 \gamma ^2+4 \gamma -8\right) }{4 (\gamma -1) \gamma  (\gamma -2)^2 (\gamma +2)} \\
    \int_0^\infty \frac{\partial^2}{\partial\gamma^2} \psi(x,\theta)\diff F_\theta(x) &=\frac{2\left(\gamma ^2+\gamma -1\right)}{ (\gamma -1)^2 \gamma ^4\Gamma(2/\gamma )} G_{3,3}^{3,2}\left(1\left|
\begin{array}{c}
 -2,-\frac{2}{\gamma },-1 \\
 -2,-2,0 \\
\end{array}
\right.\right)\sigma\\
&\phantom{{}={}}+\frac{(2 \gamma ^5-11 \gamma ^4+4 \gamma ^3+10 \gamma ^2+20 \gamma -16)}{2 (\gamma -1)^2 \gamma ^2(\gamma -2)^3 (\gamma +2)}\sigma,
\end{align*}
where $G_{3,3}^{2,3}$ denotes the Mejier-$G$ function. As all previous expressions are finite for all $\theta\in\Theta$, we conclude the assertion.
\end{proof}

\begin{figure}[!htp]
    \centering
    \includegraphics[width=\linewidth]{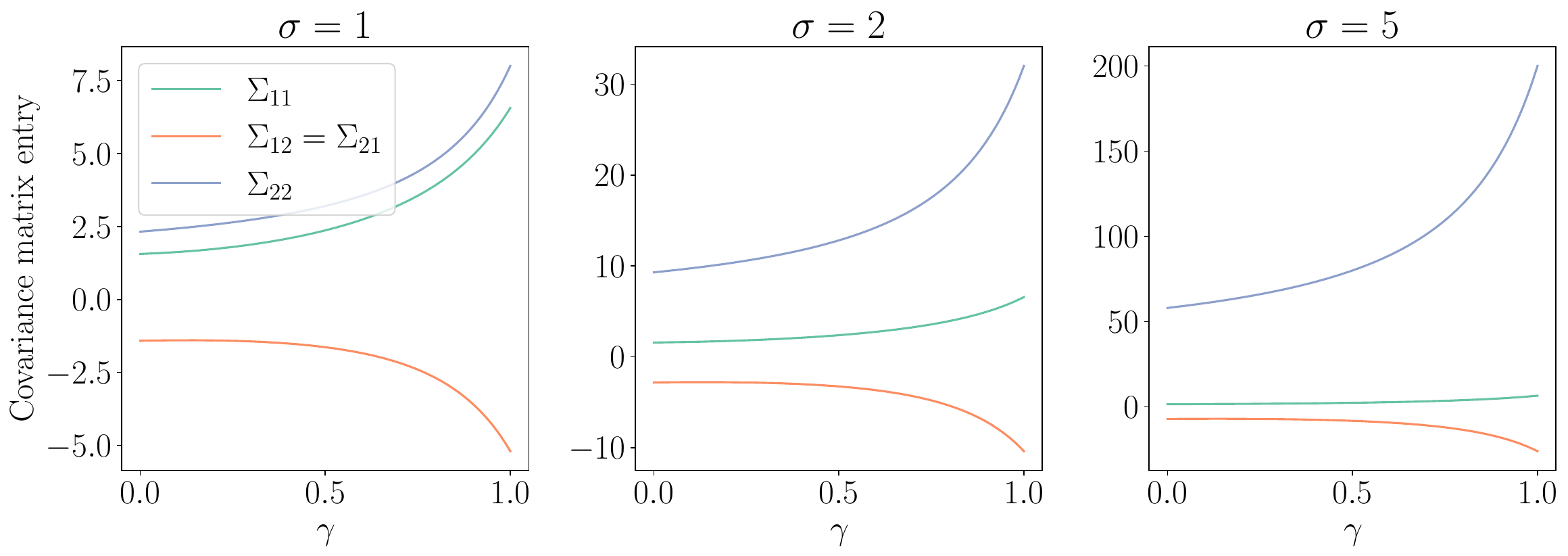}
    \caption{Entries of the covariance matrix $\gamma\mapsto\Sigma_\theta$ for $\sigma=1,2,5$, corresponding to the display \eqref{eq:covvals}. The absolute values $(\Sigma_\theta)_{ij}$ grow with $\gamma$ and $\sigma$.}
    \label{fig:cov_entries}
\end{figure}

\section{Efficiency comparison to MLE}\label{sec:efficiency}
A common estimator of the GPD parameters in the peaks-over-threshold method is the maximum-likelihood estimator (MLE). We therefore treat this well-known estimator as a benchmark to compare our MDE against. By standard assumptions, the MLE $\hat \theta^\mathrm{MLE}$ is asymptotically normal with limiting covariance being equal to the inverse Fisher information,
\begin{align*}
    \sqrt{k}(\hat \theta^\mathrm{MLE}_k-\theta_0)\xlongrightarrow{\mathcal{D}} \mathcal{N}(0,\Sigma_{\theta_0}^\mathrm{MLE}) 
\end{align*}
with
\begin{align*}
    \Sigma_\theta^\mathrm{MLE}= \left(
\begin{array}{cc}
 (\gamma +1)^2 & -((\gamma +1) \sigma ) \\
 -((\gamma +1) \sigma ) & 2 (\gamma +1) \sigma ^2
\end{array}
\right).
\end{align*}
Consequently, evaluating $\varsigma^{2,\mathrm{MLE}}_\theta(x)= \nabla_\theta F_\theta(x)^\top \Sigma_\theta \nabla_\theta F_\theta(x)$, we see

\begin{align*}
    \gamma ^4 (\sigma +\gamma  x)^2\varsigma^{2,\mathrm{MLE}}_\theta(x)&= (\gamma +1) \left(\frac{\gamma  x}{\sigma }+1\right)^{-2/\gamma }\bigg[(\gamma +1) (2 \gamma +1) \gamma ^2 x^2+(\sigma +\gamma  x) \log \left(\frac{\gamma  x}{\sigma }+1\right)\\ 
    &\phantom{=}\cdot \Big\{((\gamma +1) (\sigma +\gamma  x) \log \left(\frac{\gamma  x}{\sigma }+1\right)-2 \gamma  (2 \gamma +1) x\Big\}\bigg].
\end{align*}

By the Cram\'er-Rao-bound \cite[Pages 479-480]{cramer1946mathematical}, we cannot expect to achieve $\varsigma^{2}_\theta(x)<\varsigma^{2,\mathrm{MLE}}_\theta(x)$, but we can nevertheless compare these estimators in terms of asymptotic variance. Investigating Figure \ref{fig:eff_var}, we observe only small differences in limiting variance between MLE and MDE. This is more illustrative in the relative efficiency  $\varsigma^{2}_\theta(x)/\varsigma^{2,\mathrm{MLE}}_\theta(x)$, which appears to be well below 2 for all choices of $\gamma$, see Figure \ref{fig:eff_rel}. Of particular interest is the limiting behavior $\displaystyle \lim_{x\to 0} \varsigma^{2,\mathrm{MLE}}_\theta(x)/\varsigma^{2}_\theta(x)$ and $\displaystyle \lim_{x\to \infty} \varsigma^{2,\mathrm{MLE}}_\theta(x)/\varsigma^{2}_\theta(x)$, as
\begin{align*}
    \lim_{x\to 0} \varsigma^{2,\mathrm{MLE}}_\theta(x)/\varsigma^{2}_\theta(x)&=\frac{2 \left(8 \gamma ^5-84 \gamma ^4+374 \gamma ^3-843 \gamma ^2+944 \gamma -423\right)}{3 (\gamma -3)^2 (\gamma +1) (2 \gamma -3)^3} \\
    \lim_{x\to \infty} \varsigma^{2,\mathrm{MLE}}_\theta(x)/\varsigma^{2}_\theta(x)&=-\frac{4 (\gamma -2)^4 \left(\gamma  \left(\gamma  \left(2 \gamma  \left(4 \gamma ^2-58 \gamma +243\right)-683\right)+452\right)-639\right)}{3 (3-2 \gamma )^3 (\gamma -6)^2 (\gamma -3)^2 (\gamma +1)^2}.
\end{align*}

Notably, these limits are independent from $\sigma$ and never exceed 2 for all choices of $\gamma$, as it is visualized in Figure \ref{fig:eff_rel_limit}.

\begin{figure}[!htp]
    \centering
    \includegraphics[width=\linewidth]{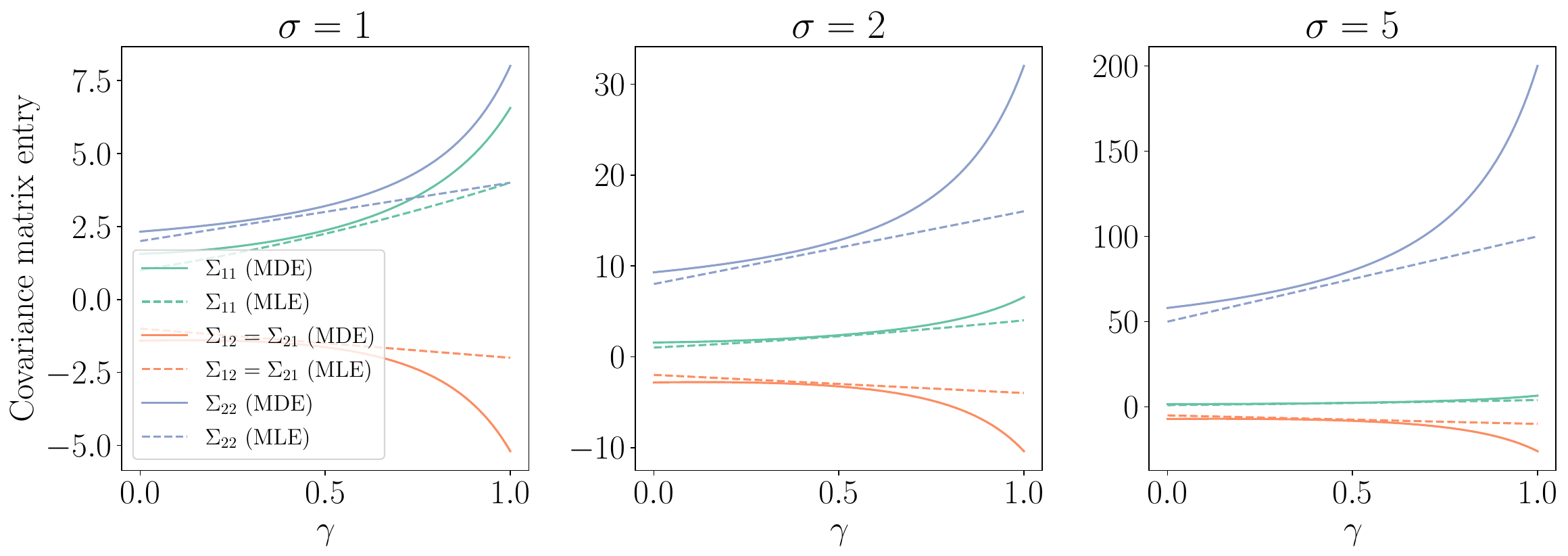}
    \caption{Entries of the limiting covariance matrices $\Sigma_\theta^\mathrm{MLE}$ and $\Sigma_\theta$ as a function of $\gamma$, shown for $\sigma=1,2,5$ (from left to right). The diagonal entries of $\Sigma_\theta^\mathrm{MLE}$ are below those of $\Sigma_\theta$.}
    \label{fig:eff_cov}
\end{figure}

\begin{figure}[!htp]
    \centering
    \includegraphics[width=\linewidth]{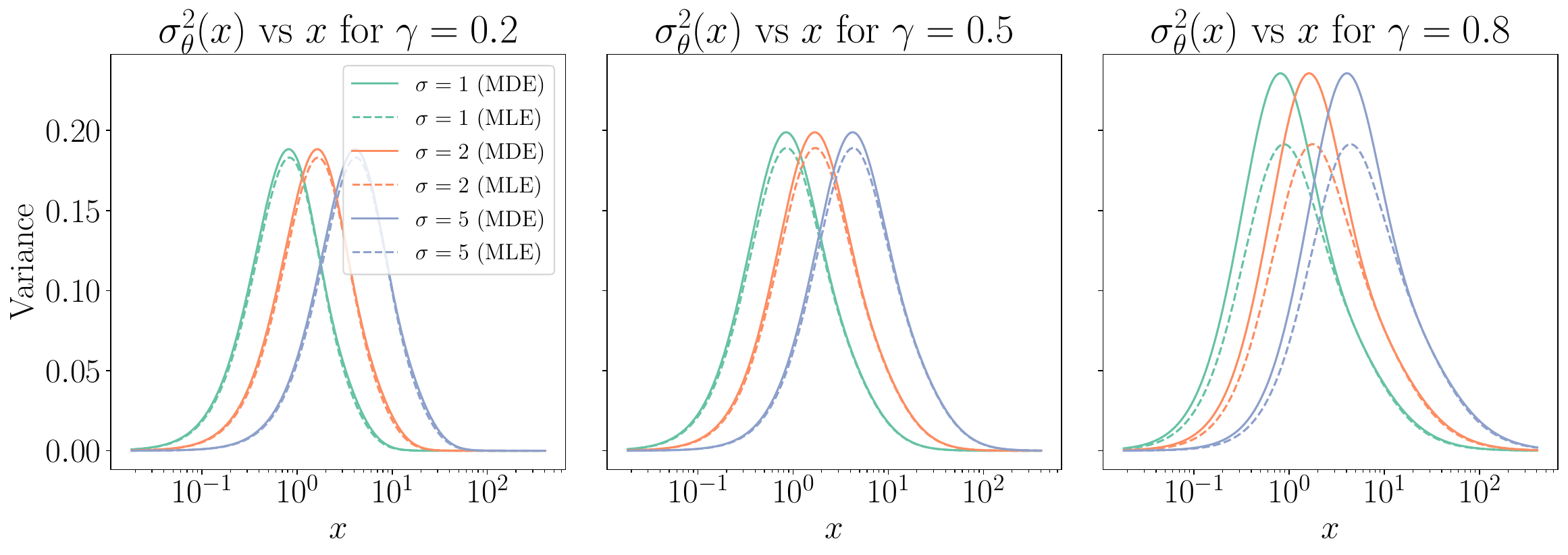}
    \caption{Variances of the survival function, $\varsigma^{2}_\theta(x),\varsigma^{2,\mathrm{MLE}}_\theta(x)$. The variances approach zero for $x\to\infty$ and $x\to 0$, with their peak depending on $\gamma$ and $\sigma$.}
    \label{fig:eff_var}
\end{figure}

\begin{figure}[!htp]
    \centering
    \includegraphics[width=\linewidth]{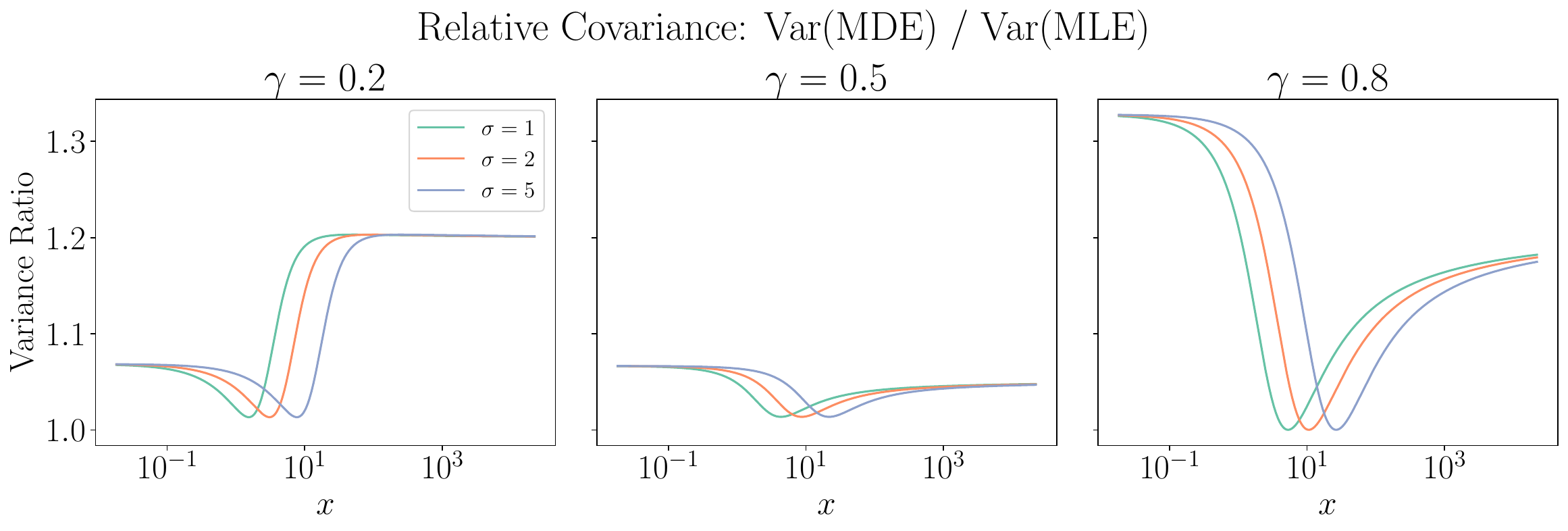}
    \caption{Relative variances of the survival function, $\varsigma^{2}_\theta(x)/\varsigma^{2,\mathrm{MLE}}_\theta(x)$ for particular choices of $\gamma=0.2,0.5,0.8$ and $\sigma=1,2,5$. Note that $\sigma$ only shifts the ratio horizontally, leaving the limits $x\to 0$ and $x\to\infty$ unaffected.}
    \label{fig:eff_rel}
\end{figure}

\begin{figure}[!htp]
    \centering
    \includegraphics[width=\linewidth]{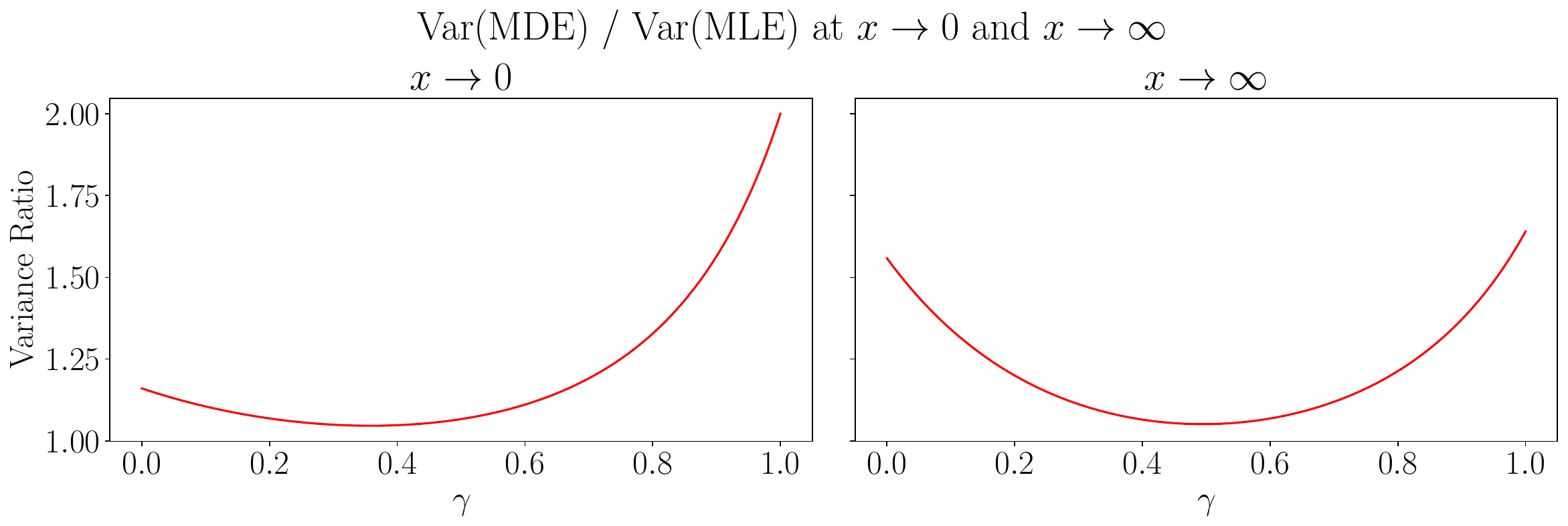}
    \caption{Limiting ratios $\displaystyle \lim_{x\to 0} \varsigma^{2,\mathrm{MLE}}_\theta(x)/\varsigma^{2}_\theta(x)$ and $\displaystyle \lim_{x\to \infty} \varsigma^{2,\mathrm{MLE}}_\theta(x)/\varsigma^{2}_\theta(x)$ as a function of $\gamma$. Note that the ratio never exceeds 2.}
    \label{fig:eff_rel_limit}
\end{figure}

\pagebreak

\section{Additional results for the challenge}\label{sec:additional}

In this section, we provide additional results on the data challenge. All plots (Figures \ref{fig:fits_p} -- \ref{fig:thchoice_p_zoom}) have their correspondence to a Figure in the main paper. For the sake of brevity, we therefore omit a lengthy discussion of each.

\begin{figure}[!htp]
    \centering
    \includegraphics[width=\linewidth]{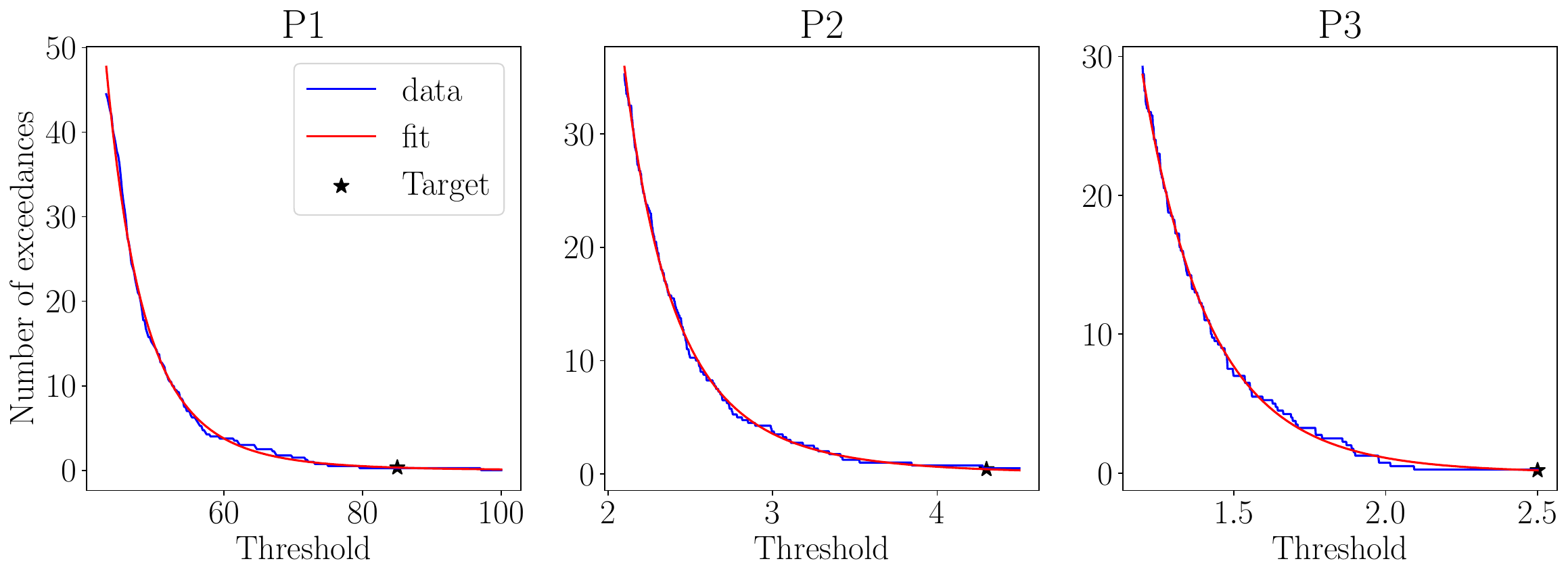}
    \caption{Average of empirical survival functions $\hat S_k(x)$ (blue) and corresponding MDE fits $\hat S_k^\mathrm{MDE}(x)$ (red), targets P1 -- P3. The black $\star$ denotes the target of the respective challenge.}
    \label{fig:fits_p}
\end{figure}

\begin{figure}[!htp]
    \centering
    \includegraphics[width=\linewidth]{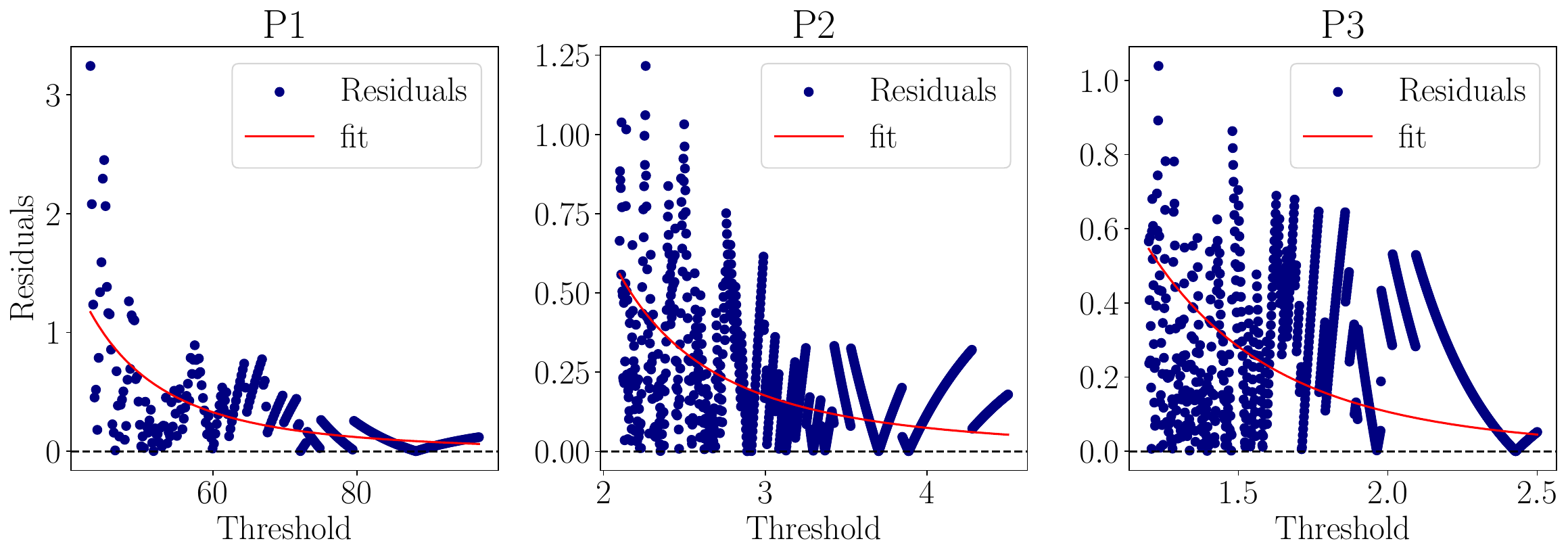}
    \caption{Residuals and fitted $\tilde\varsigma_k^2(x)$ for target P1 -- P3.}
    \label{fig:resids_p}
\end{figure}

\begin{figure}[!htp]
    \centering
    \includegraphics[width=\linewidth]{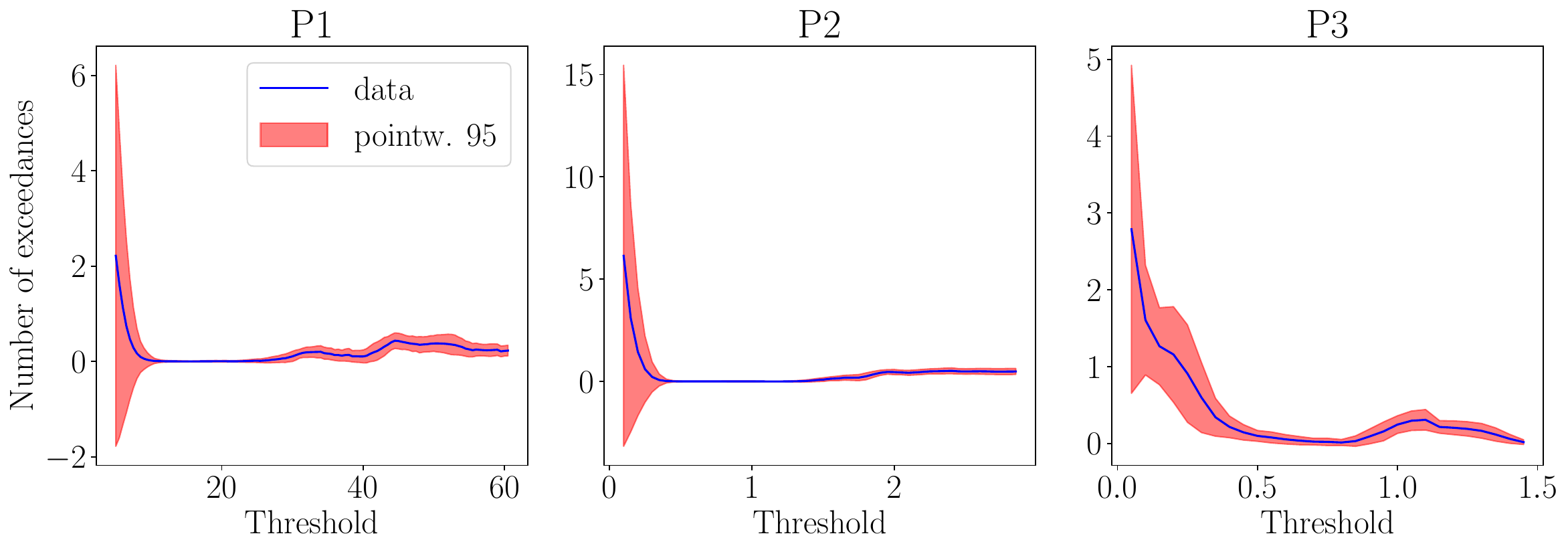}
    \caption{Development of the estimator $\hat S^\mathrm{MDE}(q)$ over the threshold $u$ for target P1 -- P3. The extreme threshold $q=x+u$ is the target threshold and remains constant.}
    \label{fig:thchoice_p}
\end{figure}

\begin{figure}[!htp]
    \centering
    \includegraphics[width=\linewidth]{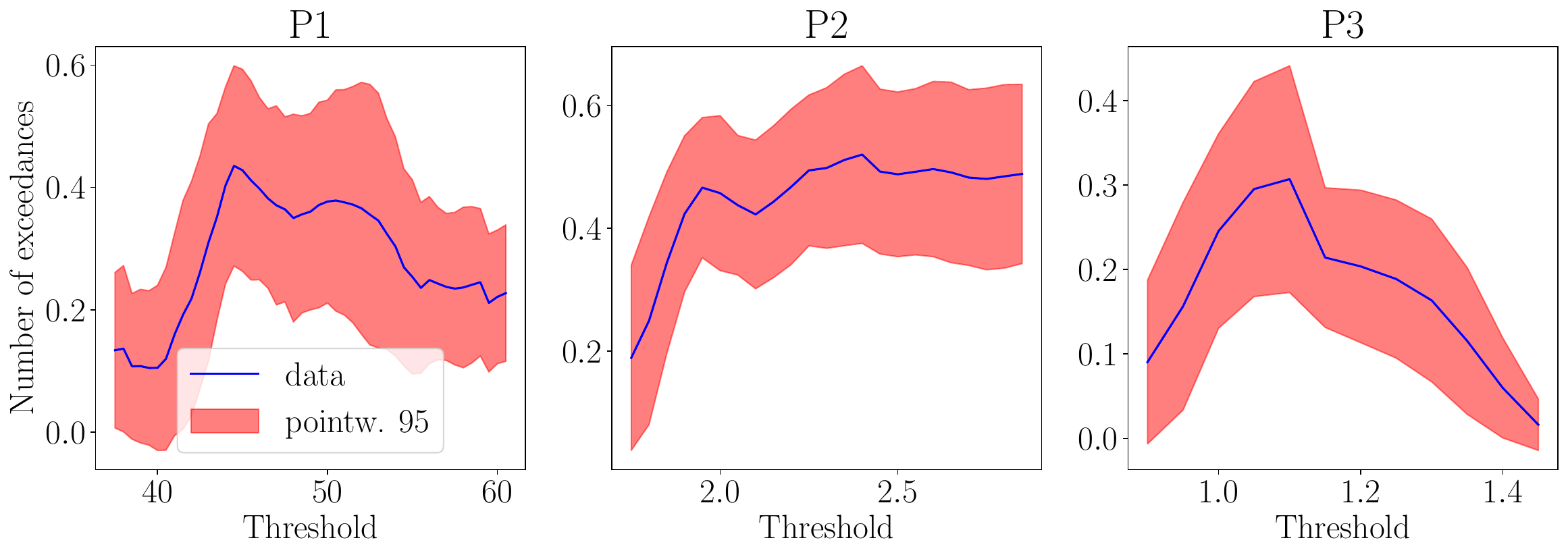}
    \caption{Zoom-in: development of the estimator $\hat S^\mathrm{MDE}(q)$ over the threshold $u$ for target C1 -- C3. The extreme threshold $q=x+u$ is the target threshold and remains constant.}
    \label{fig:thchoice_p_zoom}
\end{figure}

\pagebreak

\section{Simulation Study}\label{sec:simu}
We perform a short Monte Carlo simulation by drawing $n$ iid samples from a $\mathrm{GPD}(\gamma,1)$ distribution. In particular, we are interested in the comparison of MSE, variance and bias for different choices of $n$ and $\gamma$. Thus, we vary $n=10,50,100$ and $\gamma\in\{0.05, 0.1, ...,0.6\}$. Each MSE (variance, bias) is approximated via 1000 repetitions of each simulation experiment. The results are depicted in Figure \ref{fig:simu_gpd}. 

It is notable that all MSEs are dominated by the variance. Due to the optimality in variance, the theoretical result is confirmed that the MLE outperforms the MDE in terms of MSE for most combinations of $\gamma$ and $n$. The trends in MSE of GPD parameters directly translate to trends in MISE of survival function estimation, shown in Figure \ref{fig:simu_gpd_sf}.

\begin{figure}[!htp]
    \centering
    \includegraphics[width=.95\linewidth]{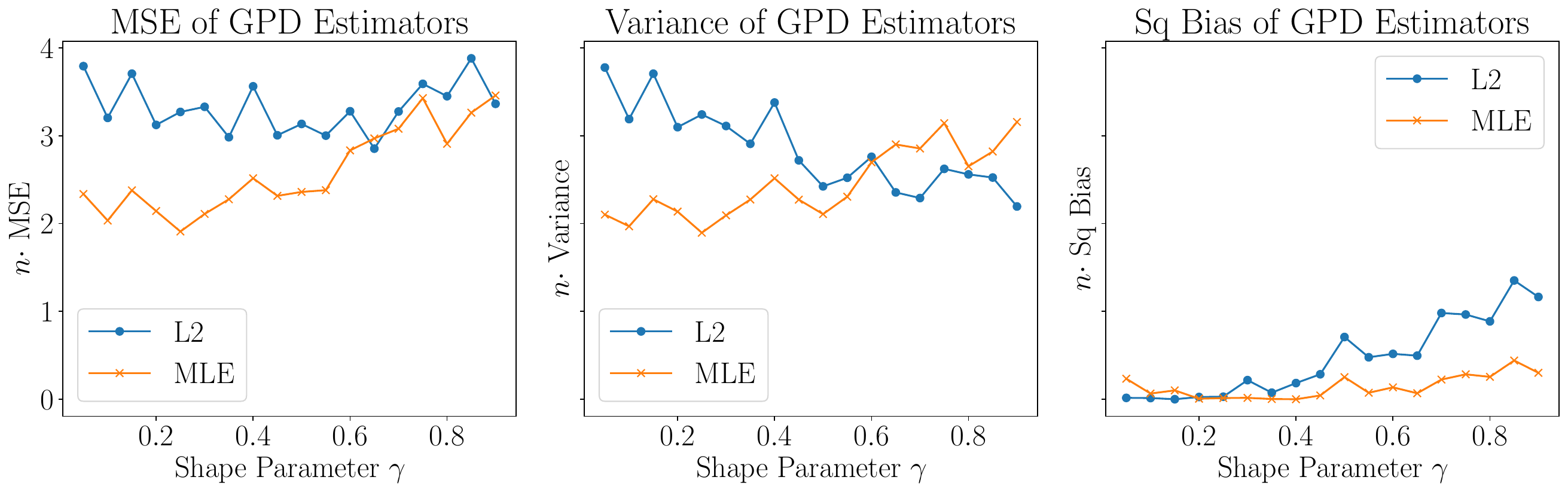}
    \includegraphics[width=.95\linewidth]{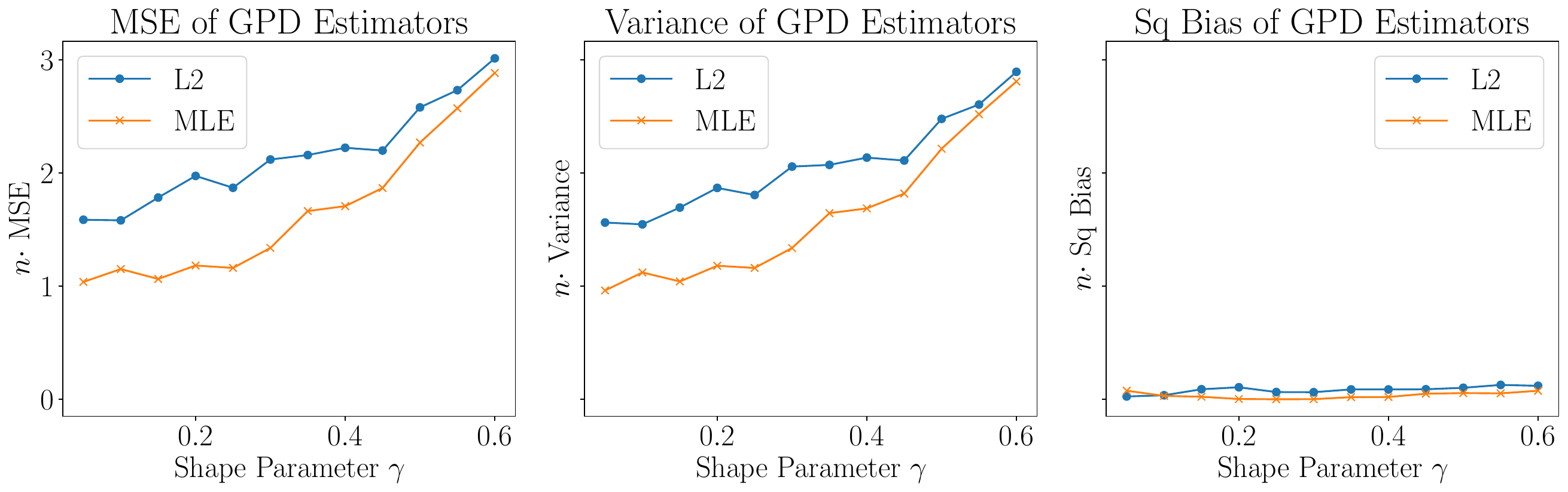}
    \includegraphics[width=.95\linewidth]{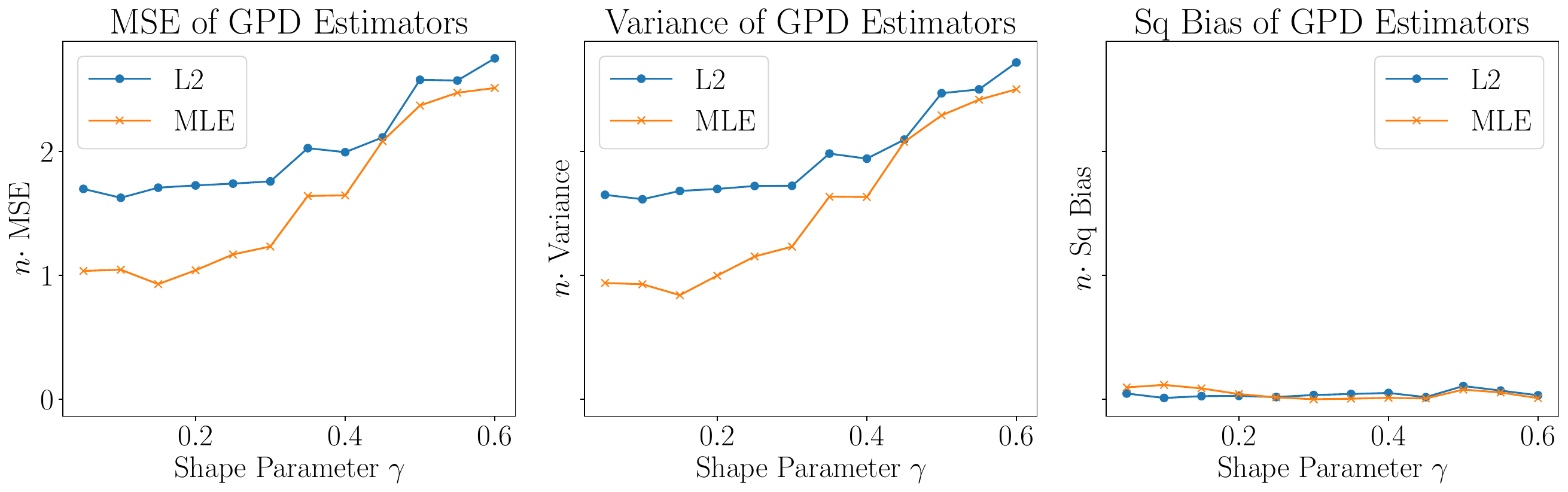}
    \caption{MSE, variance and squared bias for the shape parameter $\gamma$.  estimation. For each point, $n$ samples from a $\mathrm{GPD}(\gamma,1)$ have been drawn and the MSE was computed over 1000 repetitions. From top to bottom: $n=10,50,100$. We observe that all MSEs are variance-dominated and that MLE always outperforms MDE, but the MDE differences remain small.}
    \label{fig:simu_gpd}
\end{figure}

\begin{figure}[!htp]
    \centering
    \includegraphics[width=\linewidth]{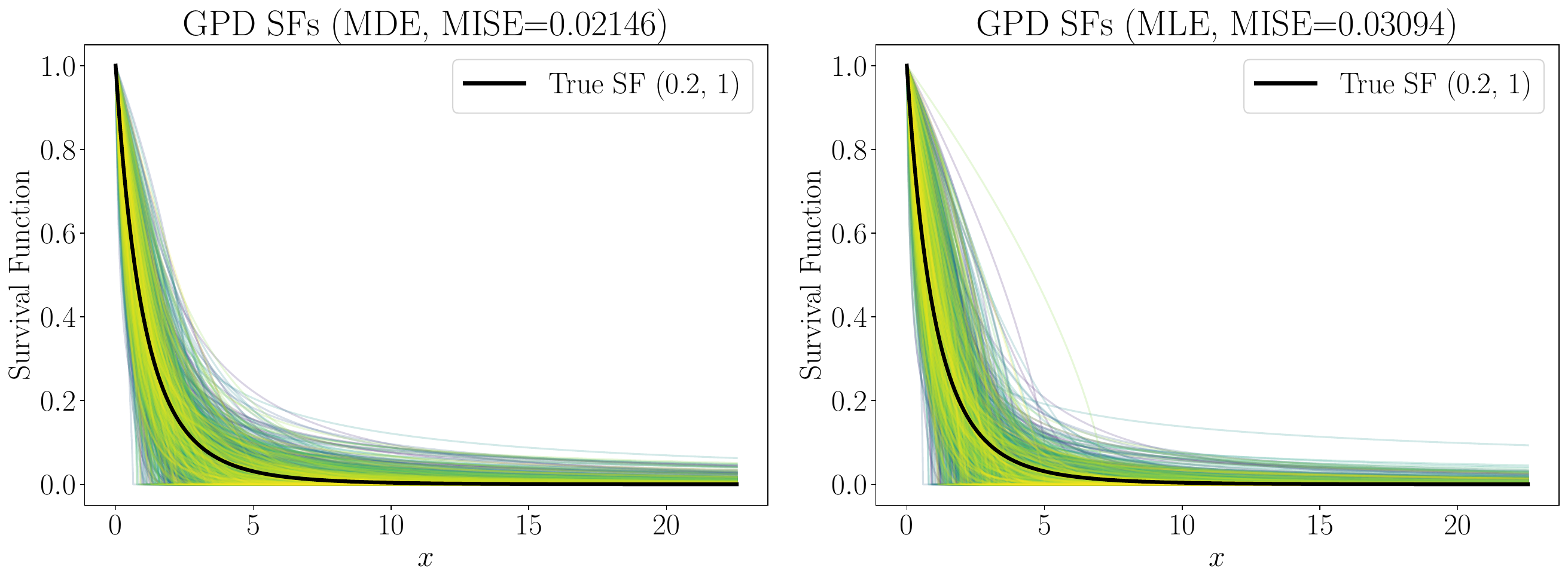}
    \includegraphics[width=\linewidth]{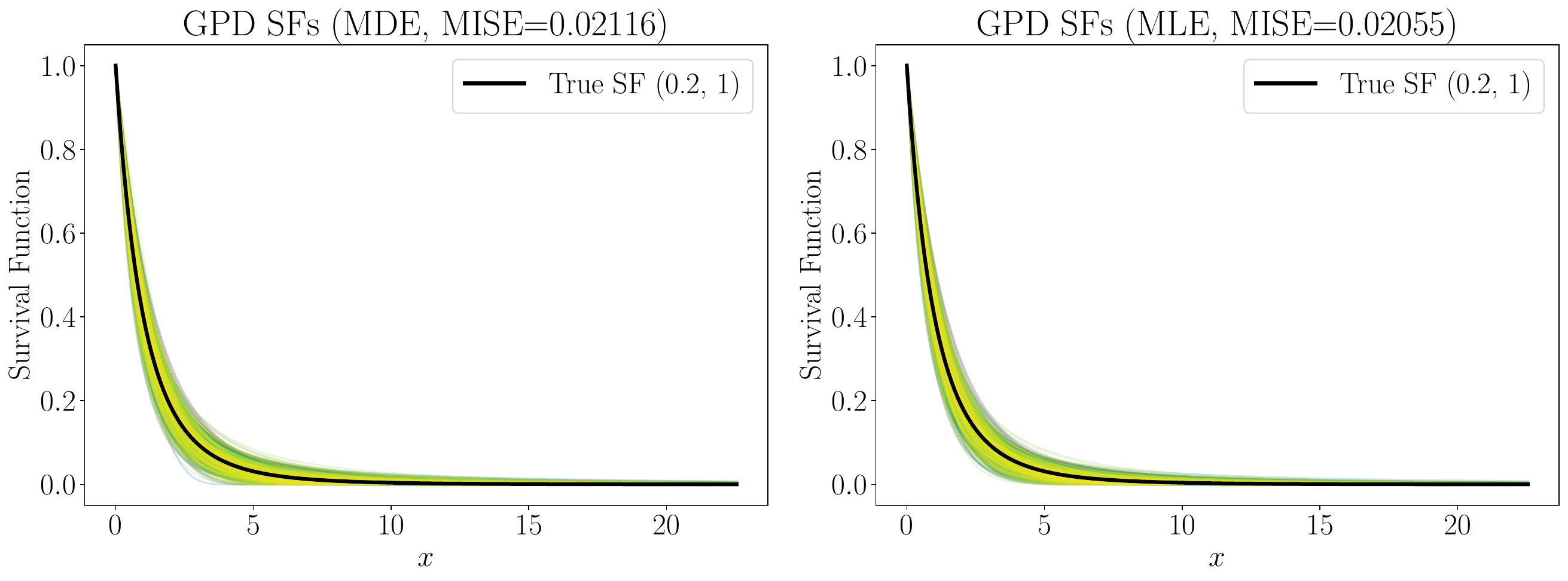}
    \caption{Estimated survival functions $S_{\hat\theta}(x)$ around true $S_{\theta_0}(x)$, $\theta_0=(0.2,1)$. Top: $n=10$, bottom: $n=100$. The MISE is computed over the shown domain on 1000 repetitions.}
    \label{fig:simu_gpd_sf}
\end{figure}

\clearpage

\newpage

\includepdf[
  pages=1,
  pagecommand={
    \section{Covariance Calculations with Mathematica}\label{app:mathematica}
  }
]{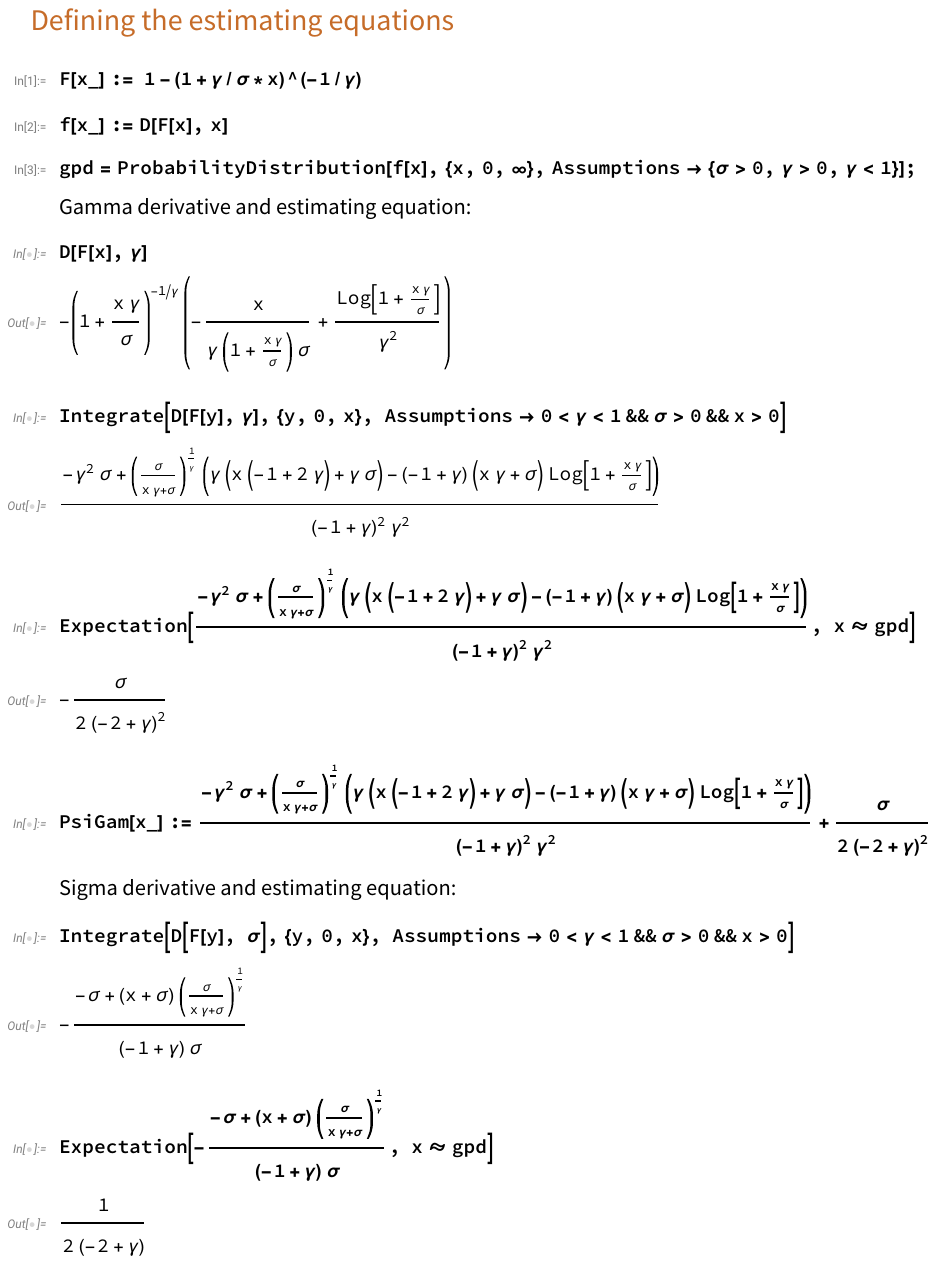}

\includepdf[
  pages=2-
]{normality.pdf}
\end{document}